\documentclass[onecolumn,12pt,letterpaper]{IEEEtran}
\usepackage{graphicx}
\usepackage{setspace}
\usepackage{cite}
\usepackage{epstopdf}
\usepackage{latexsym}
\usepackage{pifont}
\usepackage{textcomp}
\usepackage{marvosym}
\usepackage{amssymb}
\usepackage{subfigure}
\usepackage{amsmath, amsthm,  amssymb} 
\usepackage{algpseudocode}
\usepackage{algorithm}
\usepackage{hyperref}
\usepackage{wasysym}
\usepackage[english]{babel}
\usepackage{geometry}
\usepackage{color}
\begin{document}

\sloppy

\title{Token-based Function Computation with Memory}

\author{
  \IEEEauthorblockN{\normalsize Saber Salehkaleybar{\em, Student Member, IEEE}, and S. Jamaloddin Golestani{\em, Fellow, IEEE}}
  \\
  \IEEEauthorblockA{Dept. of Electrical Engineering, Sharif University of Technology, Tehran, Iran\\
    Emails: saber\_saleh@ee.sharif.edu, golestani@ieee.org} 
}

\maketitle
\begin{abstract}
In distributed function computation, each node has an initial value and the goal is to compute a function of these values in a distributed manner. In this paper, we propose a novel token-based approach to compute a wide class of target functions to which we refer as ``Token-based function Computation with Memory'' (TCM) algorithm. In this approach, node values are attached to tokens and travel across the network. Each pair of travelling tokens would coalesce when they meet, forming a token with a new value as a function of the original token values. In contrast to the Coalescing Random Walk (CRW) algorithm, where token movement is governed by random walk, meeting of tokens in our scheme is accelerated by adopting a novel chasing mechanism. We proved that, compared to the CRW algorithm, the TCM algorithm results in a reduction of time complexity by a factor of at least $\sqrt{n/\log(n)}$ in Erd\"{o}s-Renyi and complete graphs, and by a factor of $\log(n)/\log(\log(n))$ in torus networks. Simulation results show that there is at least a constant factor improvement in the message complexity of TCM algorithm in all considered topologies. Robustness of the CRW and TCM algorithms in the presence of node failure is analyzed. We show that their robustness can be improved by running multiple instances of the algorithms in parallel. 

\end{abstract}
\theoremstyle{definition}
\newtheorem{mydef}{Definition}[section]
\newtheorem{mycol}{Corollary}[section]
\newtheorem{myrem}{Remark}[section]
\theoremstyle{theorem}
\newtheorem{mylm}{Lemma}[section]
\newtheorem{myth}{Theorem}[section]
\newtheorem{myprop}{Proposition}[section]
\section{Introduction}


Distributed function computation is an essential building block in many network applications where it is required to compute a function of initial values of nodes in a distributed manner. For instance, in wireless sensor networks, distributed inference algorithms can be executed by computing average of the sensor measurements as a subroutine. Examples of distributed inference in sensor networks include transmitter localization \cite{Almodovar2012}, parameter estimation \cite{Chiuso2011}, and data aggregation \cite{Necchi2007}. As another application, consider a network with $n$ processors in which each processor has a local utility function and the goal is to obtain the optimal solution of sum of the utility functions subject to some constraints. This problem has frequently arisen in network optimization algorithms such as distributed learning \cite{mateos2010distributed}, link scheduling \cite{Lee2012}, and network utility maximization \cite{Nedic2009}. All these algorithms utilize a distributed sum or average computation subroutine in solving the optimization problems.

Consider the problem of computing a target function $f_n(v_1^0,\cdots,v_n^0)$ in a network with $n$ nodes, where $v_i^0$ is the initial value of node $i$. A common approach is based on constructing spanning trees \cite{lynch1996distributed, sappidi2013computing}. In this solution, the values would be sent toward the root where the final result is computed and sent back to all nodes over the spanning tree. Although the spanning tree-based solution is quite efficient in terms of message and time complexities, it is not robust against network perturbations such as node failures or time-varying topologies. For example, the final result may be dramatically corrupted if a node close to the root fails.

To overcome the above drawback of spanning tree-based solutions, recent approaches take advantage of local interactions between nodes \cite{Dimakis2010b}. In these approaches, each node $i$ which has a value, chooses one of its neighbors, say node $j$; The two nodes then update their values based on a predefined rule function $g(.,.)$ which is determined by the target function $f_n(.)$ (see Lemma \ref{lemma2_1}). By iterating this process in the entire network, the target function is computed in a distributed manner. Let $v_i$ and $v_j$ be the current values of nodes $i$ and $j$, respectively. Two possible options for executing the rule function $g(v_i,v_j)$ are:

\begin{equation}
\begin{cases}
    1) v_i^+=v_j^+=g(v_i,v_j),\\
    2) v_i^+=e, v_j^+=g(v_i,v_j),
\end{cases}
\label{equpdate}
\end{equation}
where $v_i^+$ and $v_j^+$ are the updated values of nodes $i$ and $j$, respectively. The value $e$ is the identity element of the rule function $g(.,.)$, i.e. $g(v,e)=g(e,v)=v$ for any value $v$.

The first option in (\ref{equpdate}) corresponds to the class of distributed algorithms commonly called gossip algorithms \cite{Dimakis2010b}. The main advantage of these algorithms is that they are robust against network perturbations due to their simple structure. However, this robust structure is obtained at the expense of huge time and message complexities \cite{Dimakis2010b}. For the first option, various updating rule functions have been proposed for specific target functions like average \cite{Boyd2006}, min/max, and sum \cite{Ayaso2010}. For instance, the updating rules $g(v_i,v_j)=(v_i+v_j)/2$ and $g(v_i,v_j)=\min(v_i,v_j)$ can be used to compute average and min functions, respectively.

The second updating option can compute a wide class of target functions including the ones computable by gossip algorithms (see Lemma \ref{lemma2_1}) and it is much more energy-efficient than the gossip algorithms \cite{saligrama2011token}. This approach can be easily implemented by a token-based algorithm: Suppose that each node has a token at the beginning of the algorithm and passes its initial value to its token. A node is said to be inactive when it does not have a token.
If the local clock of an active node like $i$ ticks, it chooses a random neighbor node, like node $j$, and sends its token carrying its value. Upon receiving the token, node $j$ updates its value, and becomes active (if it is not already)\footnote{In case of computing the sum function, the updating rule function $g(v_i,v_j)$ is $v_i+v_j$ and the identity element is equal to zero.}. Then, node $i$ sets its own value to $e$, and becomes inactive. From token's view, each token walks in the network, randomly, until it meets another token. The two tokens will then coalesce and form a token with an updated value. This process continues until the result is aggregated in one token. Finally, the last active node can broadcast the result by a controlled flooding mechanism\footnote{In section II, we will explain how the last active node broadcasts the final result.}. This computation scheme is called Coalescing Random Walk (CRW) algorithm after the coalescing random walks \cite{cooper2013coalescing}.
 
The CRW algorithm offers comparable performance to spanning tree-based solutions in terms of message complexity \cite{saligrama2011token}, making it much more energy-efficient than the gossip algorithms. However, it is still slow due to deficiency in token coalescence when only a few tokens remain in the network. Hence, authors in \cite{saligrama2011token}, modified the CRW algorithm in order to improve its running time. In the modified algorithm, which we call the truncated CRW algorithm, at some point of time, the execution of the CRW algorithm is terminated and each active node broadcasts the value of its token via a controlled flooding mechanism, leaving the completion of the computation to each network node. However, this solution does not lead to a significant improvement in time or message complexity \cite{saligrama2011token}. 

In this paper, we propose a mechanism to speed up the coalescence of tokens. Suppose that each token has a unique identifier (UID) besides its carried value. In the proposed mechanism, each node registers the maximum UID of tokens seen so far, and the outgoing edge taken by the token with the maximum UID. When a token enters a node previously visited by a token with higher UID, it follows the registered outgoing edge. Otherwise, it will go to a random chosen neighbor node, according to a predefined probability. Figure \ref{fig10} illustrates a scenario where two tokens are left in the network and show how coalescing is expedited in the proposed scheme.
Since nodes {\em memorize} the outgoing edge of a token with maximum UID they have seen, we call the proposed scheme ``Token-based function Computation with Memory'' (TCM) Algorithm.

It is interesting to mention an analogy between this scheme and cosmology. Think of tokens in the network as cosmic dusts in space. Accordingly, the process of function computation is like forming a planet from cosmic dusts. By running the TCM algorithm, tokens with small UID (light dusts) are trapped in the set of nodes visited by tokens with higher UID (in the gravitational field of heavy dusts). The coalescing process continues until a single token is left, similar to birth of a planet. 

\begin{figure}[!t]
\centering
\includegraphics[width=2.2in]{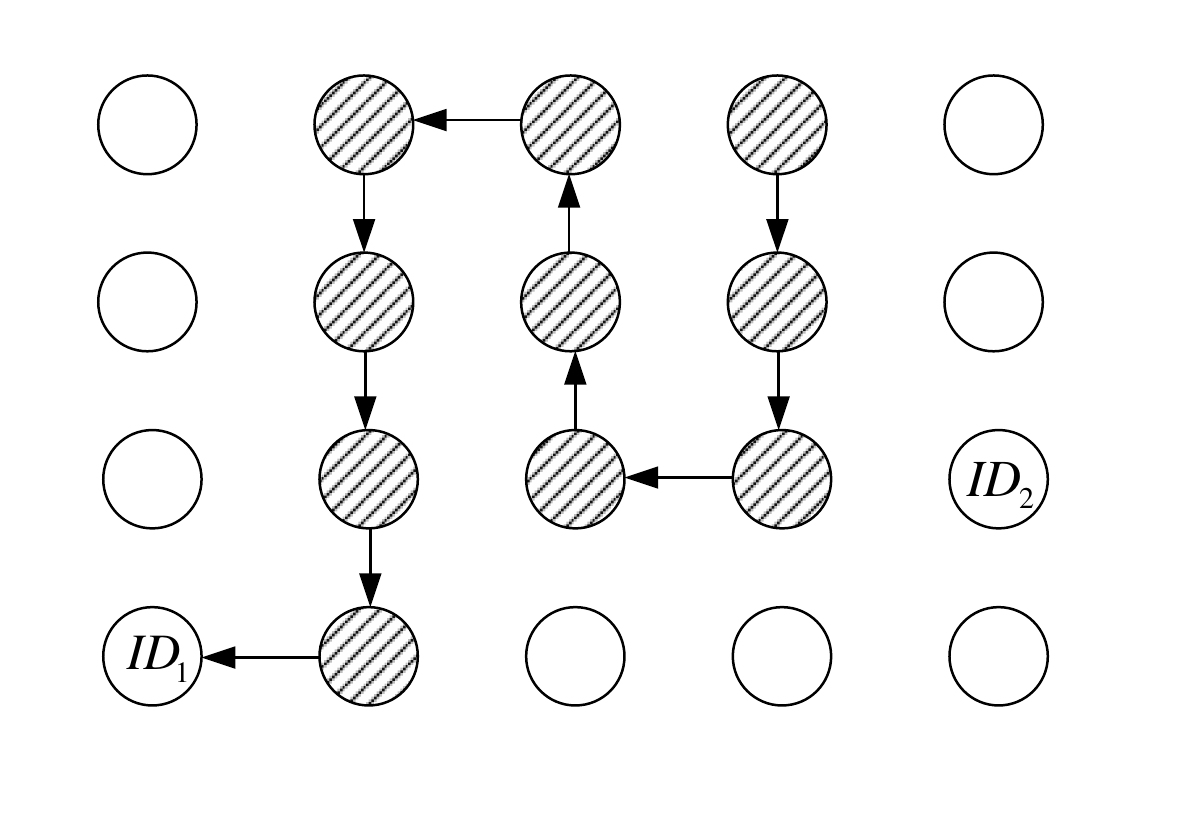}
\caption{An example of execution of TCM algorithm in a torus network: Suppose that two tokens are left in the network. Let the UIDs of the two tokens be $ID_1$ and $ID_2$ where $ID_1>ID_2$. Nodes with shaded patterns are the nodes that token $ID_1$ has visited seen so far. The arrows show the most recent direction taken by token $ID_1$. If token $ID_2$ chooses its left neighbor node in the next step, it is trapped in the set of shaded nodes and follows a path to token $ID_1$.  }
\label{fig10}
\end{figure}

The main contributions of the paper are as follows:
\begin{itemize}
\item We show that the proposed TCM algorithm, by accelerating coalescing of tokens, reduces the average time complexity by a factor $\sqrt{n/\log(n)}$ in complete graphs and Erd\"{o}s-Renyi model compared to the CRW algorithm and its truncated version. Furthermore, there is at least $\log(n)/\log(\log(n))$ factor improvement in torus networks. Simulation results show that the TCM algorithm also outperforms the CRW algorithm in terms of message complexity.

\item In CRW and TCM algorithms, the final result may be corrupted if an active node fails. Hence, it is quite important to study the robustness of these algorithms under node failures. In this regard, we evaluate the performance of CRW and TCM algorithms based on a proposed robustness metric. We show that the robustness can be substantially improved by running multiple instances of the TCM and CRW algorithms in parallel. We prove that, for the CRW algorithm, the required number of instances in order to tolerate the failure rate $\alpha/n$ in complete graphs, is of the order $O(n^{\alpha})$. While the TCM algorithm needs to run only $O(1)$ instances in parallel.

\item We study the performance of TCM and CRW algorithms under random walk mobility model \cite{camp2002survey}. Simulation results show that both algorithms can compute the class of target functions defined in Lemma II.1 successfully even in high mobility conditions.
\end{itemize}

The remainder of the paper is organized as follows: In Section II, the TCM algorithm is described. In Section III, the performances of TCM and CRW algorithms are analyzed and compared for different network topologies. In Section IV, we study the robustness of both algorithms in complete graphs. In Section V, the performances of TCM and CRW algorithms are evaluated through simulations and then compared with analytical results. Finally, we conclude with Section VI.


\section{The TCM algorithm}
\subsection{System model}
Consider a network of $n$ nodes, where each node $i$ has an initial value $v_i^0$ and the goal is to compute a function $f_n(v_1^0,\cdots,v_n^0)$ of initial values in a distributed manner. The topology of the network is represented by a bidirected graph, $G=(V,E)$, with the vertex set $V=\{1,...,n\}$, and the edge set $E\subseteq V\times V$, such that $(i,j)\in E$ if and only if nodes $i$ and $j$ can communicate directly. We index ports of node $i$ with $\{1,\cdots,d_i\}$, where $d_i$ is the degree of node $i$. 

It is assumed that the function $f_n(.)$ is symmetric for any permutation $\pi$ of the set $\{1,\cdots,n\}$, i.e. $f_n(v_1^0,\cdots,v_n^0)=f_n(v_{\pi_1}^0,\cdots,v_{\pi_n}^0)$. This means that it does not matter which node of the network holds which part of the initial values. 

\subsection{Description of the TCM algorithm}
Assume that a UID is assigned to each node $i$.\footnote{One can use randomized algorithms to assign UIDs. Each node randomly chooses an integer number in the set $\{1,\cdots,kn^2\}$. From birthday problem \cite{johnson1977urn}, it can be shown that each node gets a UID with high probability if $k$ is large enough. Furthermore, each node can encode its UID with $O(\log(n))$ bits.} At the beginning of the algorithm, each node has a token to which it passes its UID and initial value. It is also assumed that each node has an independent clock which ticks according to a Poisson process with rate one. Let the value and UID of the token at node $i$ be $value(i)$ and $ID(i)$, respectively. We denote the token at node $i$ by the vector $[value(i),size(i),ID(i)]$. The role of parameter $size(i)$ will be explained in the next part. 

The TCM algorithm computes the target function $f_n(.)$ by passing and merging tokens in the network. When a node does not have a token, it becomes inactive until a neighbor node gets in contact with it. 
Let $memory(i)$ be the maximum UID of the tokens, node $i$ has seen so far. Algorithm \ref{TCM} describes how and when an active node $i$ sends or merges tokens. The subroutine \textsc{Send()} is executed by each tick of local clock while the subroutine \textsc{Receive()} is activated upon receiving a token from some neighbor node. 

Suppose that the local clock of active node $i$ ticks. Node $i$ decides to send the token $[value(i),size(i),ID(i)]$ to a neighbor node.  In this respect, we make distinction between two cases:

Case 1- $memory(i)=ID(i)$: In this case, node $i$ decides to pass the token to a random neighbor node with probability $p_{send}$. Thus, node $i$ waits for $\frac{1}{p_{send}}$ number of clock ticks on average before sending out the token. To implement the waiting mechanism, node $i$ will exit the subroutine Send() with probability $1-p_{send}$, each time its clock ticks (line 6). Otherwise, it chooses a random port $j$, sets the $path(i)$ to $j$, and sends the token on that port (lines 7-8). 

\begin{algorithm}[!t]
\caption{The TCM algorithm} 
\label{TCM} 
\begin{algorithmic}[1]
{\footnotesize
	\State {\bf Initialization}: $memory(i) \gets ID(i)$, $path(i) \gets \{\}$, $value(i)\gets v_i^0$, $size(i)\gets 1$ ,$\forall i \in\{1,\cdots,n\}$,
    \State \qquad \qquad \qquad Node $i$ generates token $[value(i),size(i),ID(i)]$.
    \medskip
	\Procedure{Send}{ }
	 \If {$ID(i)\neq 0$} \Comment{$ID(i)$: the UID of token which is now in node $i$. It is equal to zero for inactive nodes.}
    	\If {$memory(i)=ID(i)$}  
 			\State {\bf Break} with probability $1-p_{send}$. 	
    	    \State choose a port randomly like $j$.
    		\State $path(i) \gets j$ \Comment{$path(i)$: a port number of node $i$ through which the token with highest UID has passed.}	
    	\EndIf
		\State Send token $[value(i),size(i),ID(i)]$ on port $path(i)$.
		\State $ID(i) \gets 0$, \ \  $value(i) \gets e$, \ \ $size(i)\gets 0$.   	
     \EndIf
     \EndProcedure
     \medskip
	 \Procedure{Receive}{$[value,size,ID]$}
	 	\State $ID(i)\gets \max(ID(i), ID)$
    	\State $value(i) \gets g(value(i),value)$
    	\State $size(i)\gets size(i)+size$.
     	\State $memory(i)\gets \max(memory(i),ID)$  \Comment{$memory(i)$: maximum UID that node $i$ has ever seen.}
	 \EndProcedure
}
\end{algorithmic} 
\end{algorithm}

Case 2- $ID(i)<memory(i)$: In this case, node $i$ sends the token on the port $path(i)$ with probability one. 

Now, suppose that node $i$ receives a token $[value,size,ID]$. If node $i$ is inactive, then the received token remains unchanged. Otherwise, it will coalesce with the token at nodes $i$ and the token with greater UID remains in the network (line 15). Then, the parameters $value(i)$, $size(i)$, and $memory(i)$ are updated to $g(value(i),value)$, $size(i)+size$, and $\max(memory(i),ID)$, respectively (lines 16-18). The updating rule function $g(.,.)$ is determined by the target function $f_n(.)$ as explained in Lemma \ref{lemma2_1}. Furthermore, the value $e$ is the identity element of the rule function $g(.,.)$, i.e. $g(v,e)=g(e,v)=v$ for any value $v$.

From top view, each token walks randomly in the network until it enters a node visited by a token with higher UID (Case 1). Then, it follows a path to meet the token with higher UID (Case 2). We call the walking modes in the first and second cases the {\em random walk} and {\em chasing} modes, respectively. In the random walk mode, a token walks with the lower speed $p_{send}$. Thus, it can be followed by tokens with lower UID more quickly. 

\subsection{Termination of the TCM algorithm}  The process in Algorithm \ref{TCM} continues until a few tokens remain in the network. In order to terminate the algorithm, we consider two options:
\begin{itemize}
\item Option 1- Assume that the exact network size, $n$, is known by all nodes. Furthermore, each node $i$ has a parameter $size(i)$, beside its initial value which is equal to one at the beginning. The sum of parameters $\{size(i), i\in\{1,\cdots,n\}\}$ can be computed in parallel to the target function. If the parameter $size$ in an active node reaches $n$, it can identify itself as the unique active node in the network. Then, it broadcasts the output of the TCM algorithm to all nodes by controlled flooding, further explained below. 
\item Option 2- Suppose that there exists an upper bound on the network size. Then, the execution time of the TCM algorithm can be adjusted to a time $T_{run}$ such that, on average, at most a constant number of active nodes remain after time $T_{run}$. Afterwards, each active node broadcasts the value of its token including the UID. All nodes can obtain the final result by combining values received from the active nodes. In analyzing the performances of CRW and TCM algorithms, we consider the first option.
\end{itemize}

In controlled flooding, an active node $i$ sends the value and UID of its token to all neighbor nodes. Each node $j$, upon receiving this message from a node $k$ for the first time, forwards it to all its neighbor nodes except node $k$. Since each message is transmitted on each edge at most twice, the time and message complexities of controlled flooding are $\Theta(\mbox{diam}(G))$ and $\Theta(|E|)$, respectively\footnote{In complete graphs, we can employ gossip algorithm proposed in \cite{mosk2008fast} to broadcast the output with time and message complexities of the order $O(\log(n))$ and $O(n\log(n))$, respectively.}.

The allocation of memory at node $i$ would be: $(memory(i),path(i),size(i),value(i))$ where the possible values of the first three entries are in the set $\{1,\cdots,n\}$. Thus, the TCM algorithm requires at most $\Theta(\log(n))$ bits more storage capacity compared to the CRW algorithm. The next Lemma identifies the class of target functions $f_n(v_1^0,\cdots,v_n^0)$ which can be computed by the TCM algorithm.

\begin{mylm} The TCM algorithm can compute a collection of symmetric functions $\{f_n(.)\}$ if there exists an updating rule function $g(.,.)$ such that for any permutation $\pi$ of the set $\{1,\cdots,n\}$, we have: $f_n(v_1^0,\cdots,v_n^0)=g(f_k(v_{\pi_1}^0,\cdots,v_{\pi_k}^0),f_{n-k}(v^0_{\pi_{k+1}},\cdots,v_{\pi_n}^0))$, $1\leq k\leq n$, $\forall n$.
\label{lemma2_1}
\end{mylm}
\begin{proof}
The proof is the same as Lemma 3.1 in \cite{saligrama2011token}.
\end{proof}
 A wide class of target functions fulfil these requirements such as min/max, average, sum, and exclusive OR. For instance, updating rule functions $g(v_i,v_j)=v_i+v_j$, $g(v_i,v_j)=\max(v_i,v_j)$, and $g(v_i,v_j)=v_i \oplus v_j$ are used for computing sum, minimum, and exclusive OR functions, respectively. The average function can also be computed by dividing the output of the sum function by the network size which is obtained by summing parameter $size$ of nodes in parallel to computing the sum function.

\section{Performance Analysis of the CRW and TCM Algorithms}
In this section, we study the performances of CRW and TCM algorithms in complete graphs, Erd\"{o}s-Renyi model, and torus networks. The considered network topologies may resemble different practical networks. For instance, the topology of a wireless network, in which all stations are in transmission range of each other, is typically modelled by a complete graph. A peer-to-peer network such that all nodes can communicate with each other in the overlay network, is another example of complete graphs. As we explain later, the Erd\"{o}s-Renyi model is frequently used as a model to represent social networks. Furthermore, torus network is a simple structure widely used to model distributed processing systems with grid layout or grid-based wireless sensor networks.

As a prelude to analyze the performance of the TCM algorithm, we first present an analysis of time and message complexities of the CRW algorithm for complete graphs, although the CRW algorithm is already analyzed in \cite{tavare1984line}. Then, we study time complexity of the TCM algorithm in complete graphs. We also give a naive analysis of message complexity of the TCM algorithm in complete graphs and time/message complexity of both algorithms in Erd\"{o}s-Renyi model and torus networks. The summary of time and message complexities for the TCM algorithm and the CRW/truncated CRW algorithms are given in Table 1. In complete graphs and Erd\"{o}s-Renyi model, the TCM algorithm reduces the time complexity at least by a factor $\sqrt{n/\log(n)}$. In the case of torus networks, there is an improvement at least by a factor $\log(n)/\log(\log(n))$ with respect to the CRW algorithm. Furthermore, the message complexity of the TCM algorithm is at most the same as the CRW and truncated CRW algorithms. Simulation results show that there is at least a constant factor improvement in the message complexity by employing the TCM algorithm in all considered topologies.

In analyzing the CRW and TCM algorithms, we assume that each token is transmitted instantaneously. Furthermore, passing a token is counted as sending one message in the network.

\begin{table*}
{\scriptsize
  \centering
  \caption{Performance Comparison of the TCM and CRW algorithms in terms of time and message complexities.}
  \subtable[Time complexity]{
    \begin{tabular}{|c|c|c|c|}
 		\hline
 		& Complete graphs & Erd\"{o}s-Renyi model & Torus networks
		\\
		\hline
		\hline
		TCM & $O(\sqrt{n\log(n)})$ & $O(\sqrt{n\log(n)})$ & $O(n \log(\log(n)))$
		\\
		\hline	
		CRW & $\Theta(n)$ & $\Theta(n)$ & $\Theta(n\log(n))$ \cite{saligrama2011token}
		\\
		\hline
		Truncated CRW & $\Theta(n)$ & $\Theta(n)$ & $\Theta(n)$ \cite{saligrama2011token}
		\\
		\hline
    \end{tabular}
  }
  
  \qquad\qquad\qquad\qquad\qquad\qquad\subtable[Message complexity]{
    \begin{tabular}{|c|c|c|c|}
 		\hline
 		& Complete graphs & Erd\"{o}s-Renyi model & Torus networks
		\\
		\hline
		\hline
		TCM & $O(n\log(n))$ & $O(n\log(n))$ & -
		\\
		\hline	
		CRW & $\Theta(n\log(n))$ & $\Theta(n\log(n))$ & $\Theta(n\log^2(n))$ \cite{saligrama2011token}
		\\
		\hline
		Truncated CRW & $\Theta(n\log(n))$ & $\Theta(n\log(n))$ & $\Theta(n\log^2(n))$ \cite{saligrama2011token}
		\\
		\hline
    \end{tabular}
  }
}
\end{table*}
\subsection{Time and message complexities of the CRW algorithm on complete graphs}
Let $T_{CRW}$ and $M_{CRW}$ be the average time and message complexities of the CRW algorithms, respectively. Next theorem gives a tight bound on $T_{CRW}$ and $M_{CRW}$.

\begin{myth} The average time and message complexities of the CRW algorithm in complete graphs are of the orders $\Theta(n)$ and $\Theta(n\log(n))$, respectively.
\label{th3_1}
\end{myth}
\begin{proof}
We can represent the process of token coalescing by a Markov chain with the number of active nodes remaining in the network defined as the state (see Fig. \ref{figmarkov}). The chain undergoes transition from state $k$ to state $k-1$ if a token chooses an active nodes for the next step, which occurs with rate $\frac{k(k-1)}{n-1}$. Let $T_k$ be the sojourn time in state $k$. Then the average time complexity is:

\begin{equation}
T_{CRW}=\displaystyle\sum_{k=2}^n \mathbb{E}\{T_k\}= \displaystyle\sum_{k=2}^n \frac{n-1}{k(k-1)}= (n-1)(1-1/n)\approx n-2.
\end{equation}

Besides, in state $k$, on average, $(n-1)/(k-1)$ messages are transmitted before observing a coalescing event. Therefore, the average message complexity would be\footnote{$\displaystyle\sum_{k=1}^n 1/k \approx \log(n) +c$ where $c \approx 0.577$ is the Euler-Mascheroni Constant.}:
\begin{equation}
M_{CRW}= \displaystyle\sum_{k=2}^n \frac{n-1}{k-1}\approx (n-1) (\log(n-1)+0.577).
\end{equation}
\end{proof}
\begin{figure}[!t]
\centering
\includegraphics[width=4in]{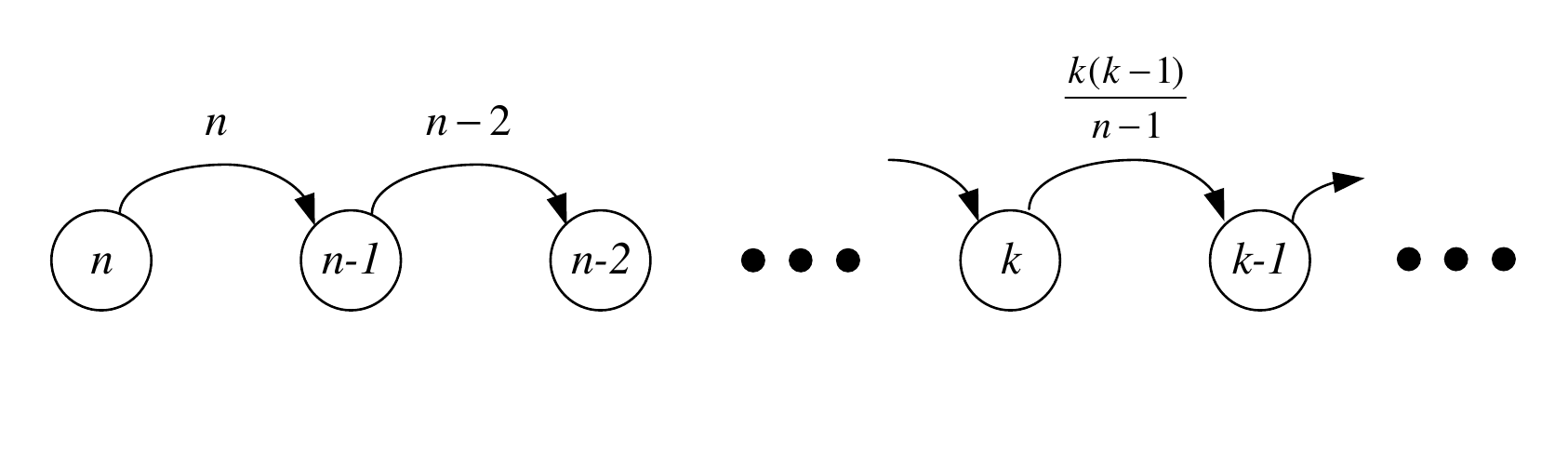}
\caption{Markov chain model for the process of token coalescing in the CRW algorithm. The state $k$ corresponds to $k$ active nodes in the network.}
\label{figmarkov}
\end{figure}
Thus, the average time and message complexities of CRW algorithm are of the orders $\Theta(n)$ and $\Theta(n\log(n))$, respectively.
\subsection{Time complexity of TCM algorithm on complete graphs}
Let the UIDs of the $n$ tokens at the beginning of the algorithm be denoted as $ID_1,\cdots,ID_n$. Without loss of generality, assume that $ID_1>\cdots >ID_n$. Throughout this section, we also assume that $p_{send}=\frac{1}{2}$.
\begin{mydef} Let $T_{coal}(ID_i)$, $i=2,\cdots,n$, denote the time that token $ID_i$ coalesces with a token with a larger UID. Thus, the algorithm running time would be: $T_{run}(n)=\max_{i \in \{2,\cdots, n\}} T_{coal}(ID_i)$.
\label{def3_1}
\end{mydef}
In the TCM algorithm, token $ID_1$ walks randomly in the network. In each step, it chooses a random node from the whole set of network nodes except the node where it is currently presented. After taking $j$ steps, the average number of visited nodes by token $ID_1$ would be: $n-(n-1)\times(1-1/(n-1))^j$. 

\begin{mydef} We call the set of nodes visited by token $ID_1$ during its first $j$ movements {\em the event horizon of $ID_1$}, and denote it by $EH1(j)$. 
\label{def3_2}
\end{mydef}

Notice that, in the TCM algorithm, when a token gets in the event horizon of token $ID_1$, it cannot escape and will eventually coalesce with token $ID_1$. We borrowed the term event horizon from general relativity, where it refers to ``the point of no return''. 

\begin{mylm} The size of event horizon of token $ID_1$ after taking $2j$ steps, i.e. $|EH1(2j)|$, is at least $\mathbb{E}\{|EH1(j)|\}\approx n(1-(1-1/n)^j)$ with probability greater than $1-e^{-n/4- j \eta}$ where constant $\eta\geq 0.05$.
\label{lemma3_2}
\end{mylm}
\begin{proof}
See Appendix A in the supplemental material.
\end{proof}

Now, we can obtain an upper bound on the average time complexity of the TCM algorithm, from Lemma \ref{lemma3_2}.

\begin{myth} In complete graphs, the average time complexity of TCM algorithm is of the order $O(\sqrt{n \log(n)})$.
\label{th3_2}
\end{myth}
\begin{proof}
For a complete proof, see Appendix B in the supplemental material. Here, in order to provide better insight about the algorithm, we present a naive analysis, that is based on a modified model of the network, where Poisson assumption for clock ticks is relaxed. Instead, we adopt a slotted model for time, where each token in the chasing mode, takes one step in each time slot. Furthermore, in the random walk mode, we replace the assumption of $p_{send}=\frac{1}{2}$ with sending token every other slot. Tokens which are scheduled to move in a time slot, take steps in a random order. 

In our analysis, we utilize the following inequality that we trust is correct, based on intuition and simulation verification:
\begin{equation}
\Pr\{T_{coal}(ID_i)\leq t\}\geq \Pr\{T_{coal}(ID_2)\leq t\}, 2\leq i\leq n.
\label{eqopen}
\end{equation}

As an example, simulation results are given for a network with $n=100$ nodes in Fig. \ref{fig8}. 

First, we derive an upper bound on the probability that the token $ID_2$ gets in the event horizon of $ID_1$ after time slot $t$. According to the simplified timing model, token $ID_1$ moves at even time slots and token $ID_2$ tries to get in the event horizon of token $ID_1$ at the same time slots. In order to obtain the upper bound, we wait for $2k$ time slots to have a big enough event horizon of token $ID_1$. Since the size of event horizon in the next $2k$ time slots is equal or greater than the one at time slot $2k$, the probability of not hitting the event horizon in time interval $[2k,4k]$ is less than $(1-|EH_1(k)|/n)^{k}$. By bounding $|EH_1(k)|$ from below (see Lemma \ref{lemma3_2}), we have for $k\geq 2\sqrt{n\log(n)}$:

\begin{equation}
\begin{split}
 \Pr\{ T_{EH1}(ID_2)>4k\}\leq & (1-\frac{ \mathbb{E}\{|EH_1(k/2)|\}}{n})^k\times \Pr\{|EH_1(k)|\leq \mathbb{E}\{|EH_1(k/2)|\}\}
 \\  &\quad+ \Pr\{|EH_1(k)|\leq \mathbb{E}\{|EH_1(k/2)|\}\} \times 1
 \\  \leq & (1-\frac{ \mathbb{E}\{|EH_1(k/2)|\}}{n})^k + e^{-n/4-\eta k/2}
 \\ \leq & e^{-\sqrt{\log(n)/n} k} + e^{-n/4-\eta k/2},
\end{split}
\label{eq1aa}
\end{equation}
where the last inequality is obtained by replacing $\mathbb{E}\{|EH_1(k/2)|\}\geq \mathbb{E}\{|EH_1(\lfloor \sqrt{n\log(n)} \rfloor )|\}$, for $k\geq 2\sqrt{n\log(n)}$.

\begin{figure}[!t]
\centering
\includegraphics[width=3.5in]{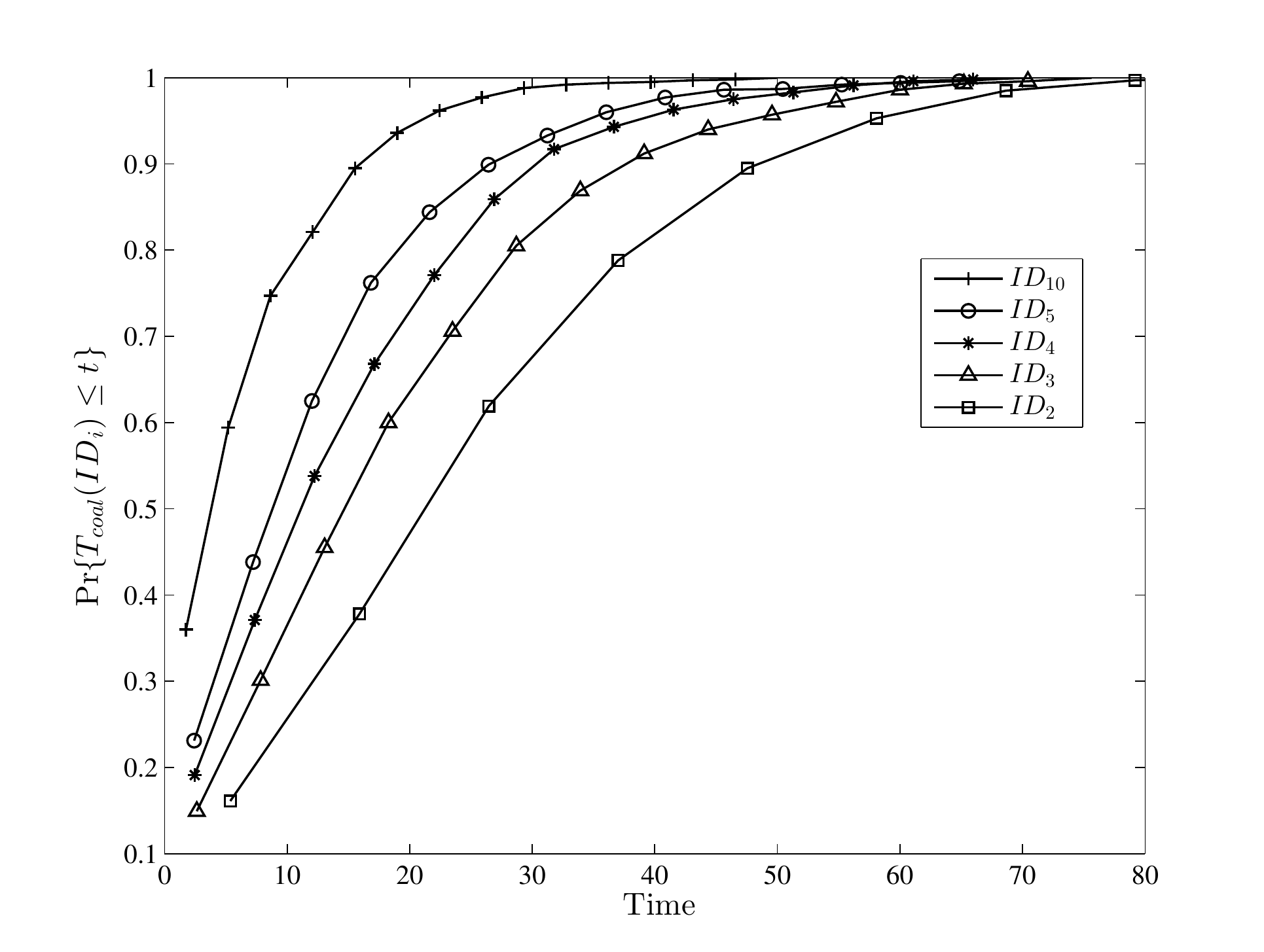}
\caption{Cumulative distribution function of coalescing time for different tokens, $n=100$.}
\label{fig8}
\end{figure}

When token $ID_2$ reaches the event horizon of token $ID_1$ at time slot $4k$, it takes at most another $4k$ time slots to coalesce with token $ID_1$. Because the size of $|EH_1(k)|$ is at most $2k$ and the relative velocity of two tokens is $1/2$. From this fact, we have: $\Pr\{T_{coal}(ID_2)\leq 8k\}\geq \Pr\{T_{EH1}(ID_2)\leq 4k\}$. From (\ref{eq1aa}), we can obtain the following:
\begin{align}
 \Pr\{T_{coal}(ID_2)>k\}&<e^{-\sqrt{\log(n)/n} k/8} + e^{-n/4-\eta k/16}, k\geq 16 \sqrt{n\log{n}} .
\label{eq1}
\end{align}
Now, an upper bound can be derived on the average time complexity:
\begin{equation}
\begin{split}
 \mathbb{E}\{T_{run}(n)\}&=\displaystyle\sum_{k=1}^{\infty} \Pr\{T_{run}(n)>k\}= \displaystyle\sum_{k=1}^{\infty} \Pr\{\max_{i \in \{2,\cdots, n\}} T_{coal}(ID_i)>k\} 
\\ &\leq \displaystyle\sum_{k=1}^{\infty} \min(1,\displaystyle\sum_{i\in \{2,\cdots,n\}} \Pr\{T_{coal}(ID_i)>k\}) 
\\& \leq^a 16 \sqrt{n\log(n)} + \int_{16 \sqrt{n\log(n)}}^{\infty} \min(1,(n-1)\times (e^{-\sqrt{\log(n)/n} t/8} + e^{-n/4-\eta t/16}))dt
\\ &\leq^b 16 \sqrt{n\log(n)} + 8/\sqrt{n\log(n)} +\frac{16n}{\eta} e^{-n/4-\eta \sqrt{n\log(n)}}.
\end{split}
\label{eq2a}
\end{equation}
\\(a) From the inequalities in (\ref{eqopen}) and (\ref{eq1}).
\\ (b) Due to $(n-1)\times (e^{-\sqrt{\log(n)/n} t/8} + e^{-n/4-\eta t/16})\leq 1$ for $t\geq 16\sqrt{n\log(n)}$. 

From (\ref{eq2a}), we conclude that the average time complexity is of the order $O(\sqrt{n\log(n)})$. Comparing with the CRW algorithm, the TCM algorithm improves the time complexity with at least a factor of $\sqrt{n/ \log(n)}$.

\end{proof}

\subsection{Message complexity of TCM algorithm on complete graphs}
\label{sec1}
In this part, we give a naive analysis of the message complexity of TCM algorithm in complete graphs. To obtain the bound on message complexity, we will show that the average number of messages sent in the TCM algorithm until observing a coalescing event, is less than the case for the CRW algorithm. 

\begin{myprop} The average message complexity of the TCM algorithm is of the order $O(n\log(n))$ in complete graphs.
\label{th3_3}
\end{myprop}
\begin{proof}
Assume that clock of an active node $i$ ticks at time $t$ and $k$ tokens remain in the network. Suppose that token $ID_r$ is in node $i$. The token $ID_r$ may be in two different modes: Walking randomly or following another token with higher UID. In the first mode, it will choose any node like $j$ with probability $1/(n-1)$. Thus, the probability of coalescing is: 
\begin{equation}
\frac{1}{n-1}\displaystyle\sum_{j\in \{1,\cdots,n\} \setminus \{i\}} \Pr\{\zeta_j(t)=1\},
\end{equation}
where $\zeta_j(t)$ is an indicator parameter which is equal to one if node $j$ is active at time $t$ and otherwise, it is zero. But the expected number of active nodes excluding node $i$ is: $\displaystyle\sum_{j\in \{1,\cdots,n\} \setminus \{i\}} 1\times \Pr\{\zeta_j(t)=1\}=k-1$. Hence, the probability of coalescing in this mode is $(k-1)/(n-1)$. 

In the second mode, token $ID_r$ follows another token with higher UID and decided to go to a neighbor node, let say node $l$. We know that there exist $k-1$ tokens excluding token $ID_r$ which walk randomly or follow another token on a trajectory of a random walk. Thus, node $l$ is active with probability at least $(k-1)/(n-1)$. Following the same arguments in analyzing the message complexity of the CRW algorithm, the message complexity is of the order $O(n\log(n))$.

\end{proof}
\subsection{Time and message complexities of TCM and CRW algorithms in Erd\"{o}s-Renyi model}
In some network applications, it is required to compute a specific function in social networks, such as majority voting \cite{Benezit2011}. Hence, it is quite important to study the performances of TCM and CRW algorithms in these scenarios. Erd\"{o}s-Renyi model is frequently used as a simple model to represent social networks \cite{newman20062}. In this part, we use this model to give a naive analysis on the time and message complexities of TCM and CRW algorithms in social networks. 

In Erd\"{o}s-Renyi model, there exists an edge between any two nodes with probability $p$. It can be shown that the graph is almost certainly connected, if $p\geq 2\log(n)/n$ \cite{erdHos1960evolution}. The next two propositions give upper bounds on the time and message complexities of CRW and TCM algorithms.

\begin{myprop} In the Erd\"{o}s-Renyi model, the average time and message complexities of CRW algorithm are of the order $O(n)$ and $O(n\log(n))$, respectively.
\label{th3_4}
\end{myprop}
\begin{proof}
Assume that $k$ tokens remain in the network. Consider token $ID_i$ walks randomly until it meets another token. In each step, it may be located in any node. From the token's view, it seems that edges are randomly established with probability $p$ in each step. Suppose that token $ID_i$ is in node $l$ at time $t$. It will choose an active node with probability, $P_{selec}$:
\begin{equation}
P_{selec}=\displaystyle\sum_{m\in \{q|\zeta_q(t)=1\}}\displaystyle\sum_{j=0}^{n-2} p\times \Pr\{d^{\prime}_l=j\} \times 1/(j+1) =(k-1)\times p \times \mathbb{E}\{1/(d^{\prime}_l+1)\},
\end{equation}
where $d^{\prime}_l$ is the degree of node $l$ excluding an active node $m$. The first term in summation shows the probability of having an edge between two nodes $l$ and $m$. The second term represents the probability that node $l$ has $j$ number of neighbor nodes excluding the node $m$ and the last term is the probability that node $l$ chooses active node $m$ from the set of its neighbor nodes. From Jensen's inequality and convexity of function $f(x)=1/(x+1)$ over $x>0$, we have: $P_{selec}\geq (k-1)p/(\mathbb{E}\{d^{\prime}_l\}+1)=(k-1)/(n-2+1/p)\geq (k-1)/(n-2+n/(2\log(n)))$. It can be easily verified that $P_{selec}\geq (k-1)/(1.12(n-1))=\Theta((k-1)/(n-1))$ for $n\geq 100$. Following the same arguments in analyzing the performance of CRW algorithm in complete graphs, we can deduce that the time and message complexities are of the order $O(n)$ and $O(n\log(n))$, respectively.
\end{proof}

\begin{myprop} In the Erd\"{o}s-Renyi model, the average time and message complexities of TCM algorithm are of the orders $O(\sqrt{n\log(n)})$ and $O(n\log(n))$, respectively.
\label{th3_5}
\end{myprop}
\begin{proof}
Suppose that the token $ID_i$ is in random walk mode. In each step, it visits each node with probability $p\times \mathbb{E}\{1/(d^{\prime}_l+1)\}\geq 1/(n-2+1/p)\approx 1/(n-1)$ for large enough $n$. Intuitively, we still have the same bounds on the probabilities $\Pr\{T_{coal}(ID_i)> t\}$, $2\leq i\leq n$. By the same arguments for the case of complete graphs, the time and message complexities are of the order $O(\sqrt{n\log(n)})$ and $O(n\log (n))$, respectively.   
\end{proof}

\subsection{Time complexity of TCM algorithm on torus networks}
In this part, we give a naive analysis on the time complexity of TCM algorithm in torus networks. We will show that the average running time of the algorithm is of the order $O(n\log(\log(n)))$. To obtain the bound, we first need to review two lemmas about single random walks. 

\begin{figure}[!t]
\centering
\includegraphics[width=1.75in]{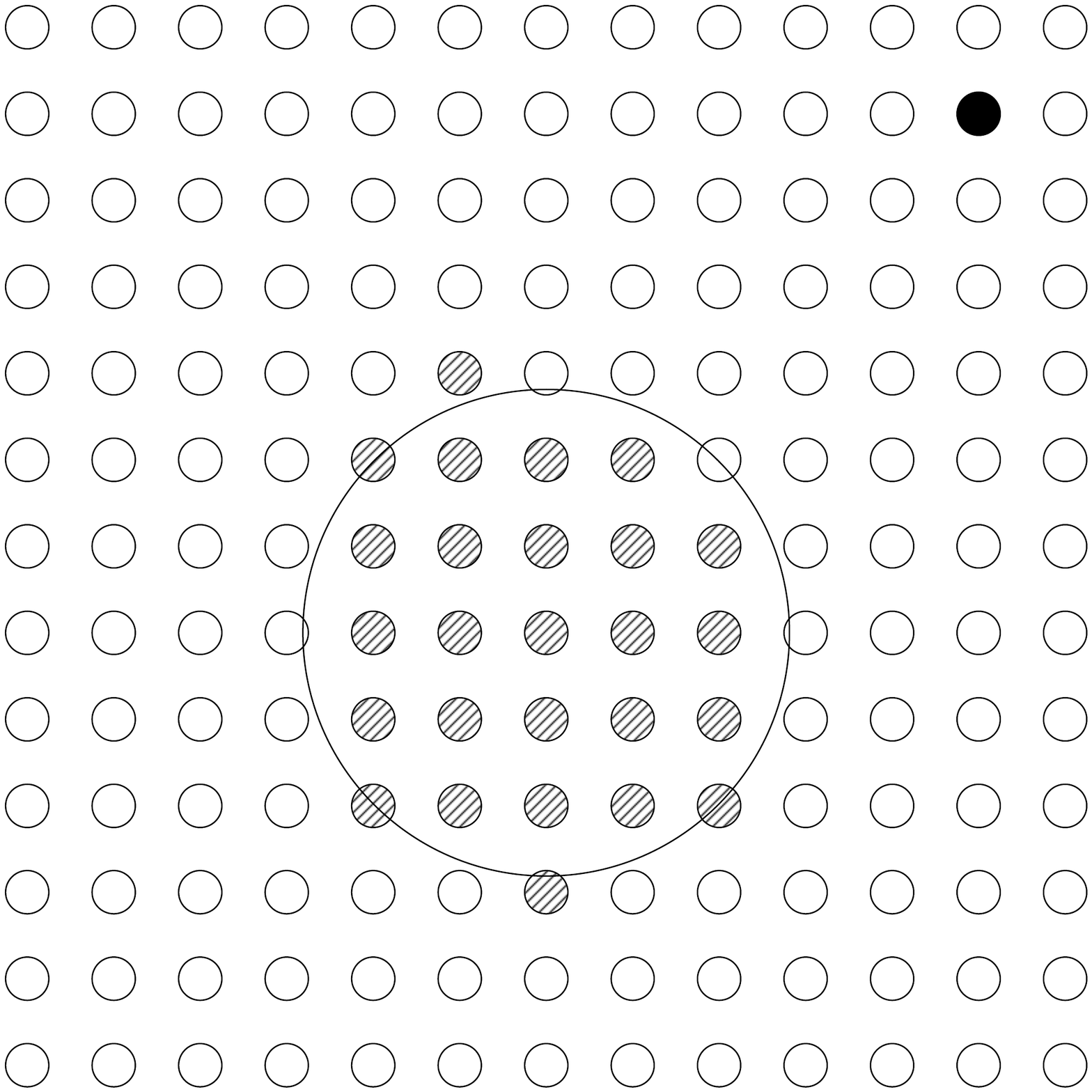}
\caption{After $k$ steps, the region of visited nodes by token $ID_1$, is approximately a disc with radius of $r=\sqrt{k /n\log(k)}$. The visited nodes are shown by dashed patterns. The average time to hit the disc by a token (depicted by black color) is $\Theta(n\log(r^{-1}))$.}
\label{fig11}
\end{figure}
\begin{mylm} \cite{dembo2004cover} Consider a $\sqrt{n}\times \sqrt{n}$ discrete torus. Let $T_{hit}$ be the average time for a single random walk to hit the set of nodes contained in a disc of radius $r<\mathcal{R}/2$ around a point $x$ starting from the boundary of a disc of radius $\mathcal{R}$ around $x$. Then, we have: $\mathbb{E}\{T_{hit}\}=\Theta(n\log(r^{-1}))$.
\label{lemma3_4}
\end{mylm}

\begin{mylm} \cite{dvoretzky1951some} Let $V_k$ be the number of nodes visited by a single random walk on $\mathbb{Z}^2$ after $k$ steps. Then, we have: $\mathbb{E}\{V_{k }\}=\frac{\pi k}{\log{k}}$ and variance $\mathrm{Var}(V_{k })=O(k^2 \frac{\log(\log(k))}{\log(k)^3})$.
\label{lemma3_5}
\end{mylm}

\begin{myprop} In torus networks, the average time complexity of the TCM algorithm is of the order $O(n\log(\log(n)))$.
\label{th3_6}
\end{myprop}
\begin{proof}
Consider the token $ID_1$. From Lemma \ref{lemma3_5}, $\frac{\pi k}{\log{k}}$ number of nodes are visited on average by token $ID_1$ after $k$ steps. To simplify the analysis, we approximate the region of visited nodes with a disc of radius $\sqrt{k/n\log{k}}$ on a unit torus (see Fig. \ref{fig11}). Hence, after $k=\beta n$ steps, radius of the disc would be $\sqrt{\beta/\log(\beta n)}$ where $\beta<<1$. Furthermore, any other token $ID_i$ ($i\geq 2$) walks randomly or follows another token on a trajectory of a random walk. Hence, from Lemma \ref{lemma3_4}, token $ID_i$ hits the disc after $\Theta (n\log(\log(n)))$ average time units if it does not coalesce with any other token during this time interval. Following that, at most $2n$ time slots are required to reach token $ID_1$. Therefore, the time complexity is of the order $O(n\log(\log(n)))$.
\end{proof}

\section{Robustness Analysis}
In this section, we study the robustness of CRW and TCM algorithms. In the literature of distributed systems, identifying robust algorithms is done mostly from a qualitative rather than quantitative perspective. For instance, there is a common belief that gossip algorithms have a robust structure against network perturbations such as node failures or time-varying topologies \cite{Dimakis2010b}. Nevertheless, this advantage is achieved by huge time and message complexities \cite{Dimakis2010b}.

To the best of our knowledge, there exist a few works \cite{gupta2006robustness,leblanc2013resilient} on analyzing the robustness of distributed function computation (DCF) algorithms. One of the main challenges is that it is difficult to devise a well defined robustness metric. Despite the challenges, there exist some methodologies for defining a robustness metric in a computing system \cite{ali2004measuring,shestak2008stochastic}. Here, we follow the same approach in these methodologies. To do so, three steps should be taken:

1) First, a metric should be considered for the system performance. In our case, we consider it as the probability of successful computation at the end of the algorithm, i.e. $\Pr \{v_i=f(v_1^0,\cdots,v_n^0), \forall i \in \{1,\cdots,n\} ,\mbox{ node $i$ has not failed} \}$ where $v_i$ is the output of node $i$. Note that the correct result is a function of initial values of whole nodes. 

2) In the second step, network perturbations should be modelled. In the CRW and TCM algorithms, the final result may be corrupted if an active node fails. Thus, studying the impact of such event on the robustness of these algorithms is quite important. In order to model node failures, we assume that each node may crash according to exponential distribution with rate $\lambda$. Therefore, the average lifespan of a node is $1/\lambda$. As a result, at most $n\times (1-e^{-\lambda \mathbb{E}\{T_{run}(n)\}})$ number of nodes fail on average. We assume that the expected number of crashed nodes during the execution of the algorithm is at most a small fraction of network size, i.e. $\lambda \mathbb{E}\{T_{run}(n)\}<-\log(1-\alpha)\approx \alpha$ where $\alpha<<1$.

3) At the end, it should be identified how much perturbation the algorithm can tolerate such that the performance metric remains in an acceptable region. For this purpose, we define the following robustness metric.

\begin{mydef} The robustness metric, $r(\epsilon)$, is defined by the following equation:
\begin{align}
\nonumber r(\epsilon)&\triangleq \max \lambda_0 \\
 s.t. \Pr \{v_i=f(v_1^0,\cdots,v_n^0), \forall i  \in \{1,\cdots &,n\} , \mbox{ {\small node $i$ has not failed}}| \lambda=\lambda_0\}\geq 1-\epsilon,
\end{align}
\label{def4_1}
\end{mydef}
Intuitively, the robustness metric shows maximum failure rate which an algorithm can tolerate such that the probability of successful computation is greater than a desired threshold, $1-\epsilon$. 
In order to execute CRW and TCM algorithms in the presence of node failure, it is assumed that each token chooses a random neighbor node for the next clock tick, if the contacting node at the current moment has been failed. 

\subsection{Robustness of CRW algorithm in complete graphs}
We first derive the probability that node $i$ is active at time $t$, i.e. $\Pr\{\zeta_i(t)=1\}$.

\begin{mylm} In the non-failure scenario, node $i$ is active at time $t$ with probability $\Pr\{\zeta_i(t)=1\}=1/(t+1)$.
\label{lemma4_1}
\end{mylm}
\begin{proof}
We use the mean field theorem to calculate the probability $p(t)=\Pr\{\zeta_i(t)=1\}$ (for more on mean field theorem, see \cite{takayasu1992extinction}). Due to symmetry property of the complete graphs, each node is active at time $t$ with the same probability $p(t)$. Thus, the portion of active nodes will decrease with rate $-p^2(t)$. Therefore, we have: $\frac{dp(t)}{dt}=-p^2(t)$. By solving the differential equation and considering the fact that $p(0)=1$, we have: $p(t)=1/(t+1)$ and $\mathbb{E}\{c(t)\}=n/(t+1)$ where $c(t)=\displaystyle\sum_{i=1}^n\zeta_i(t)$ is the the number of active nodes at time $t$.
\end{proof}

\begin{mylm} In the CRW algorithm, the probability of successful computation is greater than $n^{-\lambda n}$ for the node failure rate $\lambda <\alpha/\mathbb{E}\{T_{run}(n)\}$.
\label{lemma4_2}
\end{mylm}
\begin{proof}
The function computation is successful iff none of active nodes fail up to time $T_{run}(n)$.\footnote{In controlled flooding mechanism, the value of last active node is broadcasted to all nodes. Thus, node failures have negligible impact on the final result in this phase and we neglect it in our analysis.} 
 Let $F_{[t_0,t_1)}$ be the event that none of active nodes fails in the time interval $[t_0,t_1)$. Thus, the probability $P_{succ}(t)\triangleq \Pr\{F_{[0,t)}\},(t<T_{run}(n))$, satisfies the following equation:
\begin{align}
\nonumber P_{succ}(t+dt)&=P_{succ}(t)\times\Pr\{F_{[t,t+dt)}|F_{[0,t)}\},
\\ \nonumber &=P_{succ}(t)\times \mathbb{E}_{c(t)}\left\lbrace\Pr\{F_{[t,t+dt)}|c(t), F_{[0,t)}\}\right\rbrace,
\\ \nonumber &=^{a}P_{succ}(t) \times\mathbb{E}_{c(t)}\{ e^{-\lambda c(t)dt}\},
\\ \nonumber &=P_{succ}(t) \times \mathbb{E}_{c(t)}\{1-\lambda c(t) dt\}+ O(dt^2),
\\ \nonumber &=^{b}P_{succ}(t) \times (1-\frac{\lambda n}{t+1} dt).
\end{align}
(a) From property of exponential distribution considered in modelling node failures.
\\(b) We assume that $\mathbb{E}\{c(t)\}\approx n/(t+1)$ is not affected by missing a small fraction of nodes.

Therefore, we have:
\begin{equation}
\frac{dP_{succ}(t)}{dt}=-P_{succ}(t)\frac{\lambda n}{t+1}.
\end{equation}

By solving the above differential equation, we have: $P_{succ}(t)=(t+1)^{-\lambda n}$. Hence, we can obtain a lower bound on the probability of successful computation, $P_{succ}$, as follows:
\begin{equation}
P_{succ}= \mathbb{E}_{T_{run}(n)}\big\{P_{succ}\big(T_{run}(n)\big)\big\}\geq (\mathbb{E}\{T_{run}(n)\}+1)^{-\lambda n}\geq n^{-\lambda n}.
\label{eqrobust1}
\end{equation}
The above inequality holds due to Jensen's inequality and considering the fact that function $f(x)=(x+1)^{-n\lambda}, x>0$ is convex. 
\end{proof}

After some manipulations, it can be easily verified that: $r(\epsilon)> \log((1-\epsilon)^{-1})/(n\log (n))$. Hence, the single CRW can tolerate failure rates of order $O(1/(n\log(n)))$. But, how can we improve the performance of this algorithm such that it tolerates failure rates of order $\alpha/\mathbb{E}\{T_{run}(n)\}=\alpha/n$? One effective solution is to run multiple CRWs in parallel. More specifically, we run $R$ instances of CRW algorithm denoted by $1,\dots,R$; As a result, if an active node fails in some instances of the CRW algorithm, it might be inactive in the other instances and those instances survive from that node failure. 

In order to run multiple instances of the algorithm, tokens carry the index of the corresponding instance in the execution of the algorithm and can only coalesce with token of the same index. At the end of the algorithm, nodes decide on the output of an instance which includes as many values as possible in computing the target function. To do so, we can assume that each node $i$ has a count parameter $size(i)$ which is equal to one at the beginning of the algorithm (see section II). The sum of these count parameters is obtained alongside computing the target function of initial values for each instance of the algorithm. Nodes decide on the output of instance with maximum count parameter.

\begin{mylm} To tolerate the failure rate of $\alpha/n$ and get the correct result with probability $1-\epsilon$, the number of instances of the CRW algorithm should be greater than:
\begin{align}
R>\log(\epsilon^{-1})n^{\alpha}.
\end{align}
\label{lemma4_3}
\end{mylm}
\begin{proof} 
Assuming that the multiple instances are approximately independent and considering $\lambda= \alpha/n$ and Lemma \ref{lemma4_2}, the probability of successful computation of the target function with $R$ instances of CRW algorithm is greater than:
\begin{align}
\nonumber 1-(1-&n^{-\alpha})^R\geq 1-\epsilon,
\\ \rightarrow R \geq &\frac{\log(\epsilon)}{\log(1-n^{-\alpha})}\approx \log(\epsilon^{-1})n^{\alpha}.
\end{align}
\end{proof}
\begin{mycol}The CRW algorithm is robust against failing $\alpha$ fraction of nodes by running $O(n^\alpha)$ instances of CRW algorithm in parallel. Thus, the message complexity is of the order $O(n^{1+\alpha}\log(n))$. Since $\alpha<<1$, this solution imposes low message overhead.
\label{col4_1}
\end{mycol}
\subsection{Robustness of TCM algorithm in complete graphs}
To study the robustness of TCM algorithm, we first need to obtain the average percentage of active nodes at time $t$. However, deriving $\mathbb{E}\{c(t)\}/n$ for TCM algorithm in complete graphs is not an easy task as the one for the CRW algorithm. Since it is required to compute the following sum:
\begin{equation}
\mathbb{E}\{c(t)\}=\frac{1}{n}\displaystyle\sum_{i=1}^n\Pr\{T_{coal}(ID_i)>t\},
\end{equation}
where obtaining $\Pr\{T_{coal}(ID_i)>t\}, \forall i\in\{2,\cdots,n\}$ (or even bounds on them) is quite challenging. In order to simplify the analysis, we consider a form of function $\mathbb{E}\{c(t)\}/n\approx \log_2(t+2)/(at^2+bt+1)$ where $a=0.23$ and $b=1.8$. The reason for choosing this form is that the average running time is of the order $O(\sqrt{n\log(n)})$ and it can also be fitted properly to the simulation results\footnote{From simulation results, the root mean square error (RMSE) of fitted function is less than $10^{-3}$ for all $n\in [100,2500]$.}. According to this assumption, we can derive the probability of successful computation by the following lemma.

\begin{mylm} The probability of successful computation by TCM algorithm is greater than $ e^{-\gamma n\lambda }$ in complete graphs where $\gamma\approx 4.13$.
\label{lemma4_4}
\end{mylm}
\begin{proof}
By the same arguments in the proof of Lemma \ref{lemma4_2}, we have:
\begin{equation}
P_{succ}(t)=\exp\bigg(-\lambda \int_{0}^t \mathbb{E}\{c(\tau)\}d\tau\bigg).
\end{equation}

Since $h(t)=e^{-\lambda t}$ is convex and non-increasing and $g(t)=\int_{0}^t \mathbb{E}\{c(\tau)\}d\tau$ is concave ($\frac{d}{dt} \mathbb{E}\{c(t)\}<0, t>0$), the $P_{succ}(t)=h(g(t))$ is convex. Hence, we have from Jensen's inequality:
\begin{align}
\nonumber P_{succ}=\mathbb{E}_{T_{run}(n)}\big\{P_{succ}\big(T_{run}(n)\big)\big\}&\geq \exp\bigg(-n \lambda \int_{0}^{\mathbb{E}\{T_{run}(n)\}} \frac{\log_2(\tau+2)}{a\tau^2+b \tau+1}d\tau\bigg)
\\ &\geq e^{-\gamma n\lambda},
\label{eqrobust2}
\end{align}
where $\int_{0}^{\mathbb{E}\{T_{run}(n)\}} \frac{\log_2(\tau+2)}{a\tau^2+b \tau+1}d\tau \leq  \int_{0}^{\infty} \frac{\log_2(\tau+2)}{a\tau^2+b \tau+1}d\tau=\gamma$. 
\end{proof}

\begin{mycol} From Lemma \ref{lemma4_4}, we can see that $r(\epsilon)$ is at least $\epsilon/(\gamma n)$ for a single TCM algorithm. Similar to the CRW algorithm, we can run multiple instances of TCM algorithm in parallel to improve its robustness. In order to tolerate the failure rate of $\alpha/n$, the required number of instances running in parallel should be of the order $O(1)$.
\label{col4_2}
\end{mycol}
\begin{figure}[!t]
    \centering
    \subfigure[The TCM algorithm]
    {
        \includegraphics[width=3in]{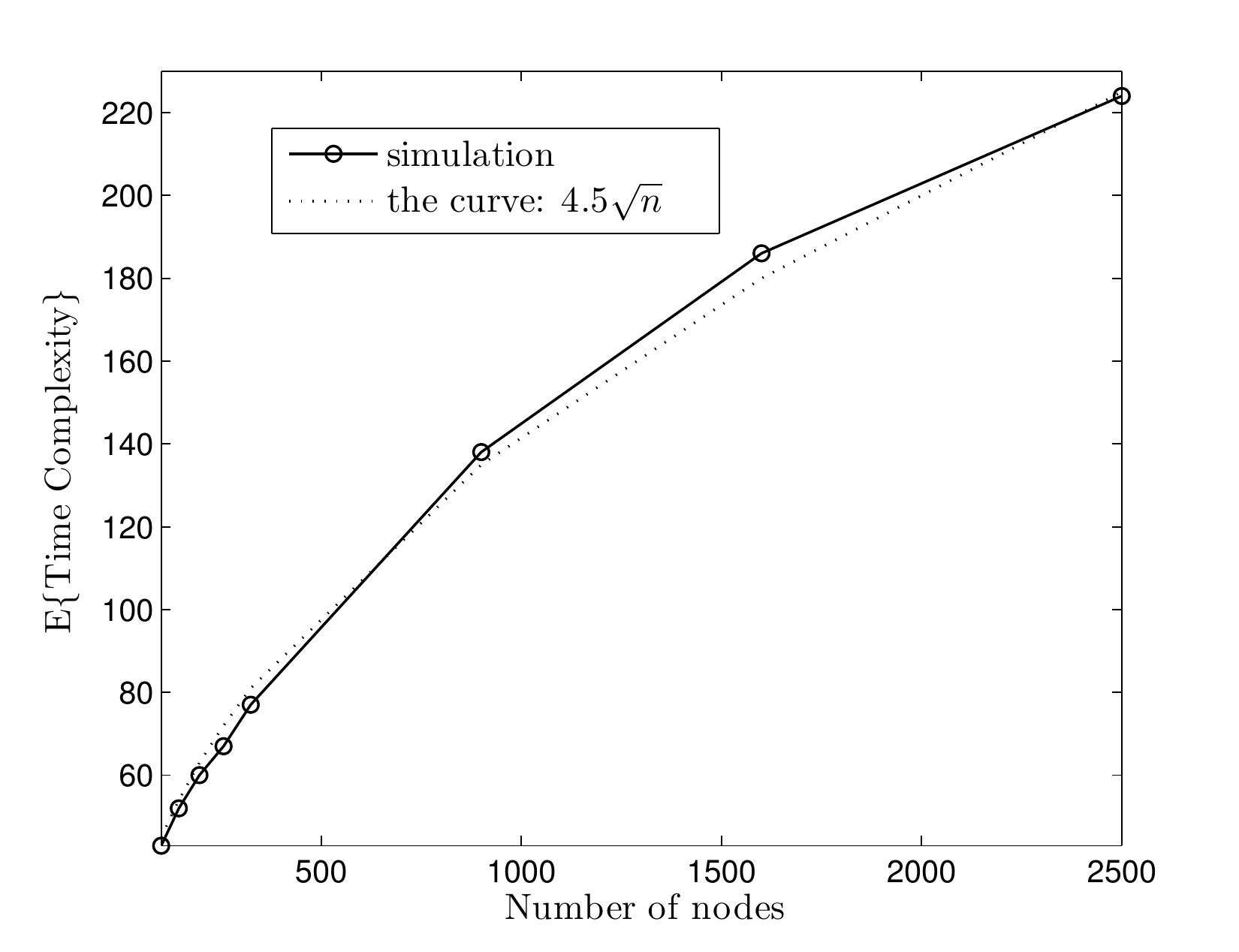}
        \label{fig1a}
    }
    \subfigure[The CRW algorithm]
    {
        \includegraphics[width=3in]{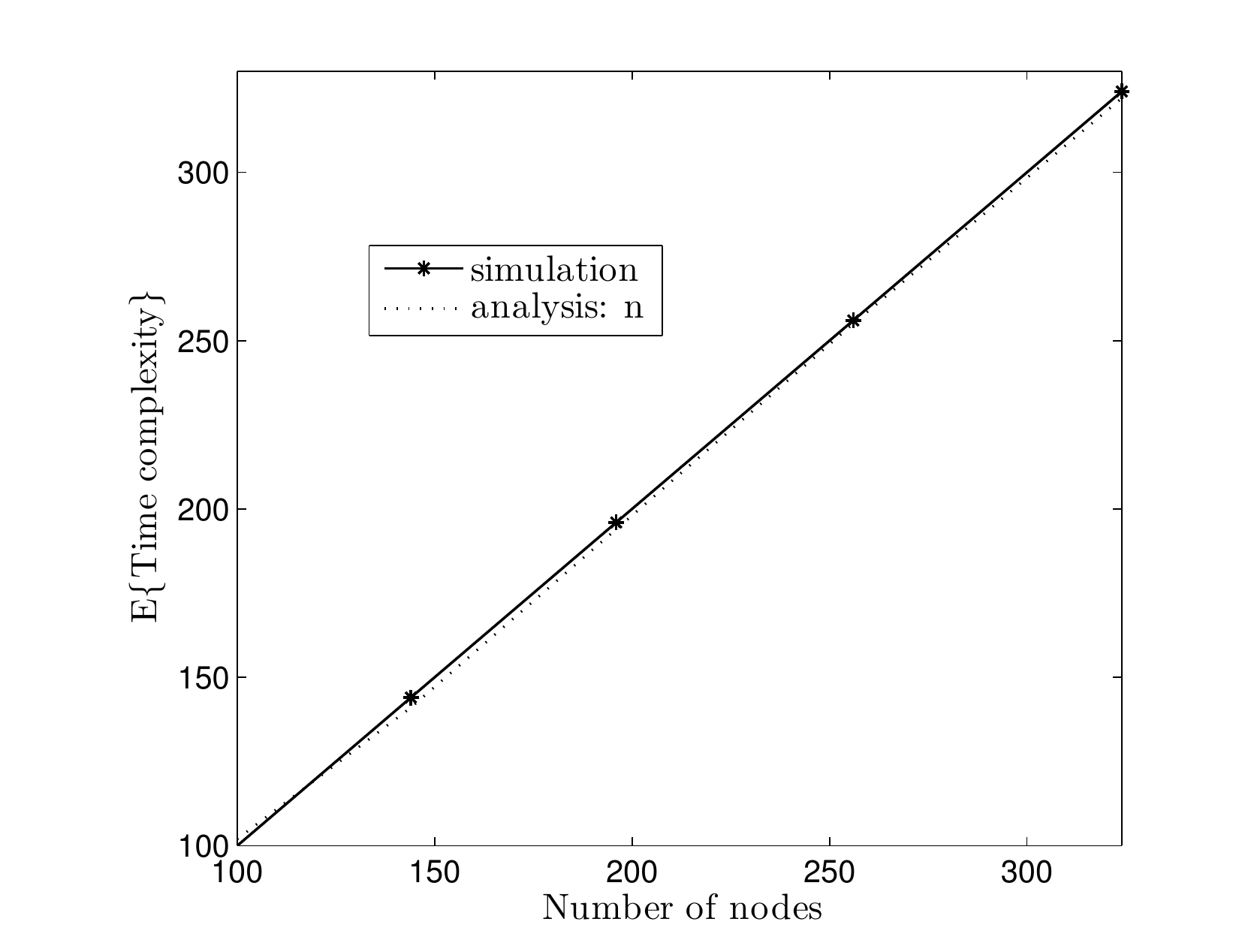}
        \label{fig1b}
    }
	 \caption{Average time complexities of TCM and CRW algorithms in complete graphs.}
\label{fig1}
\end{figure} 
\section{Simulation Results}
In this section, we evaluate the performances of TCM and CRW algorithms through simulation.
Simulation results are averaged over 10000 runs for both algorithms in complete graphs, torus networks, and Erd\"{o}s-Renyi model. 

In Fig. \ref{fig1}, average time complexities of TCM and CRW algorithms are given for complete graphs. In the TCM algorithm, $p_{send}$ is set to $\frac{1}{2}$. As it can be seen, simulation results are close to our analysis. Furthermore, the TCM algorithm outperforms the CRW algorithm by a scale factor $\sqrt{n}$. For instance, for $n=256$, the average time complexities of TCM and CRW algorithms are $67$ and $255$ time units, respectively. Hence, the amount of improvement is $255/67=3.81 \approx n/(4.5n^{0.5})=3.56 $. In Fig. \ref{fig2}, the average message complexities of TCM and CRW algorithms are depicted in complete graphs. As it can be seen, the average message complexity of TCM algorithm is always less than half of the one for the CRW algorithm. 

\begin{figure}[!t]
    \centering
    \subfigure[The TCM algorithm]
    {
        \includegraphics[width=3in]{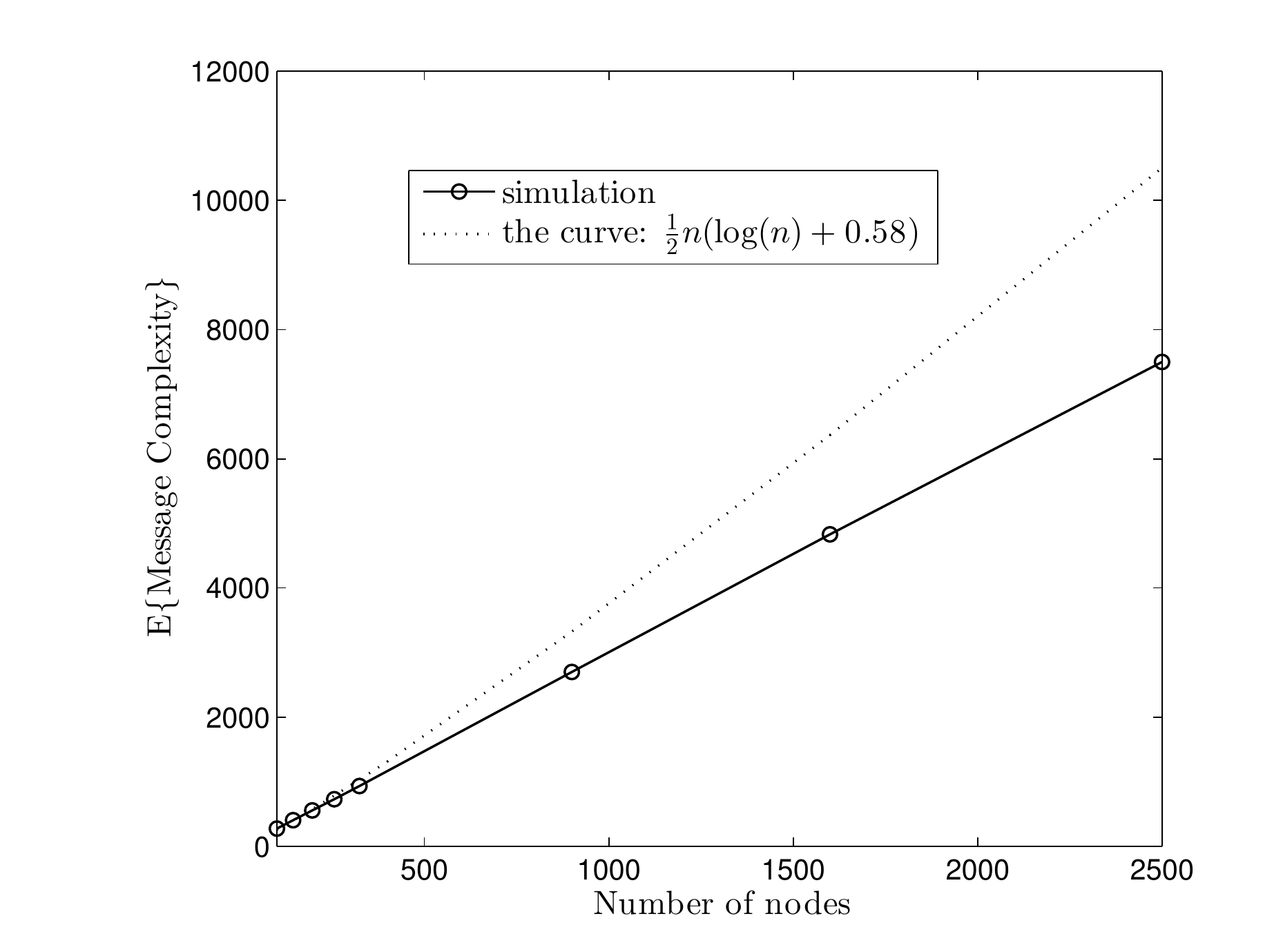}
        \label{fig2a}
    }
    \subfigure[The CRW algorithm]
    {
        \includegraphics[width=3in]{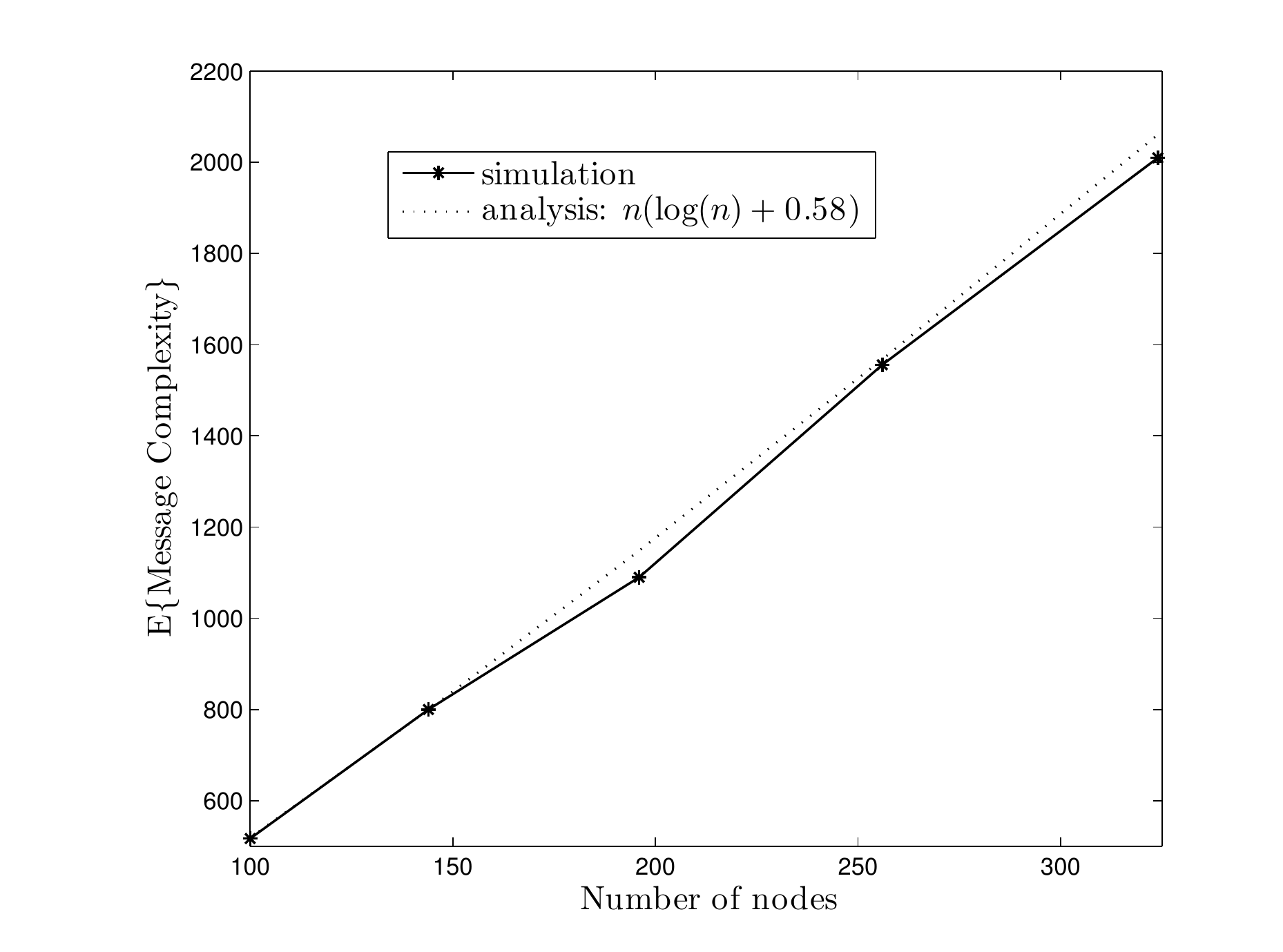}
        \label{fig2b}
    }
	 \caption{Average message complexities of TCM and CRW algorithms in complete graphs.}
\label{fig2}
\end{figure} 

In order to study the effect of parameter $p_{send}$ on the running time of TCM algorithm, the average time complexity is plotted versus $p_{send}$ for the complete graphs in Fig. \ref{figpsend}. Intuitively, the event horizon of token $ID_1$ grows with a pace inversely proportional to $p_{send}$. On the other hand, the relative velocity of two tokens is approximately related to $1-p_{send}$. Thus, the average time complexity increases as $p_{send}$ goes to zero or one. Furthermore, the optimal $p_{send}$ gets close to $0.5$ as network size increases. 
\begin{figure}[!t]
\centering
\includegraphics[width=3.6in]{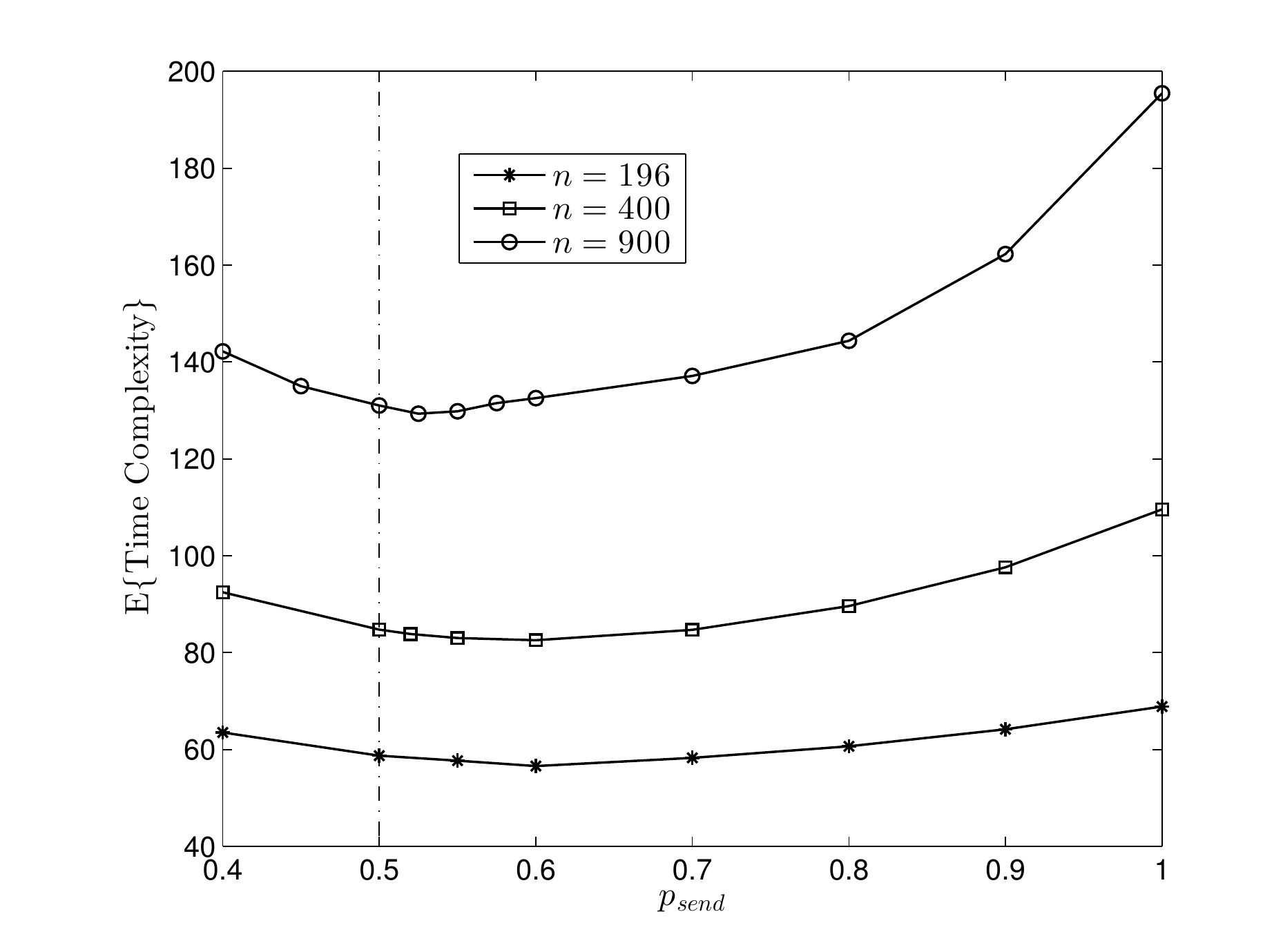}
\caption{Average time complexity of TCM algorithm versus $p_{send}$. } 
\label{figpsend}
\end{figure}
\begin{figure}[!t]
    \centering
    \subfigure[Average time complexity]
    {
        \includegraphics[width=3in]{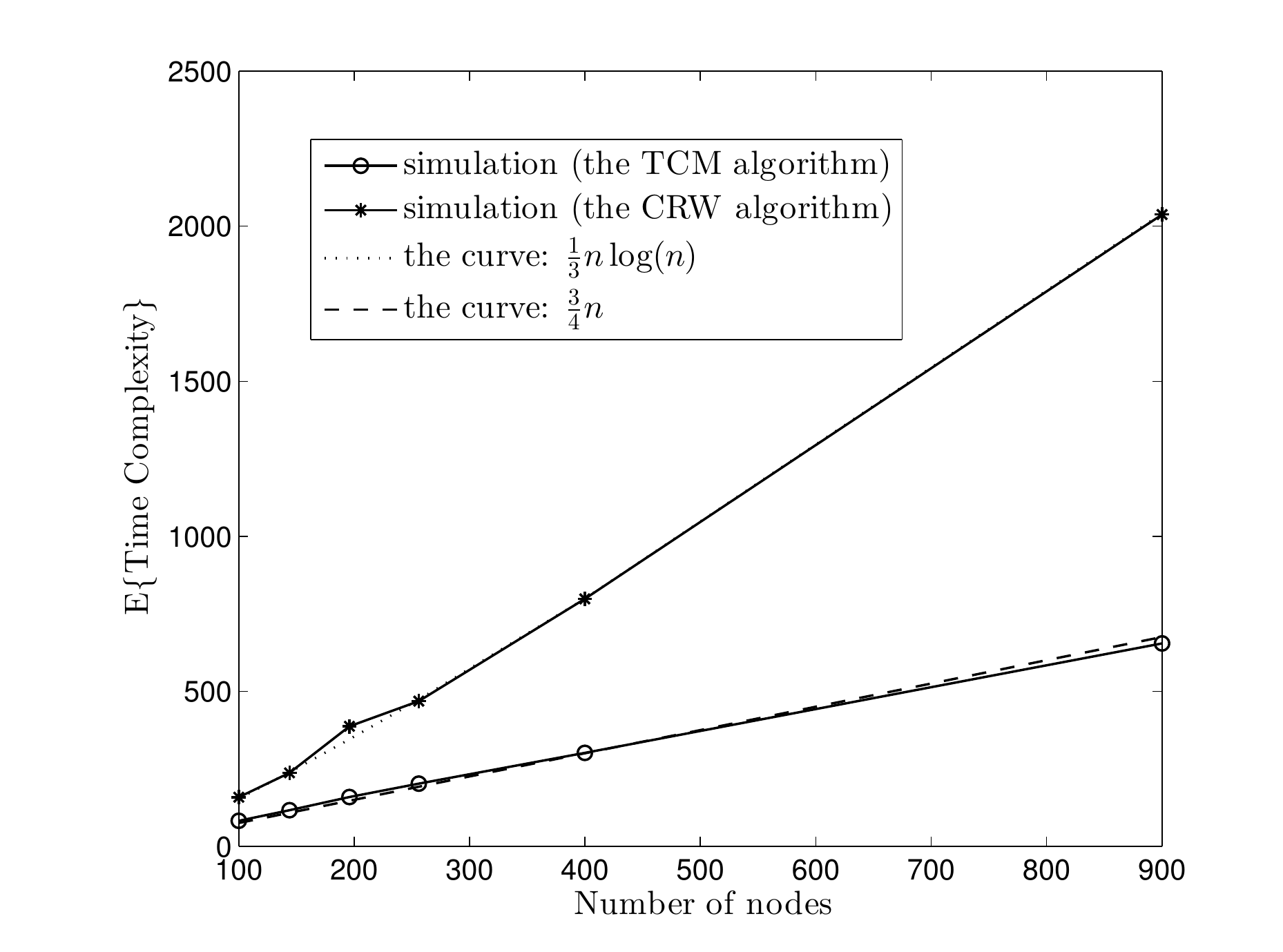}
        \label{fig5a}
    }
    \subfigure[Average message complexity]
    {
        \includegraphics[width=3in]{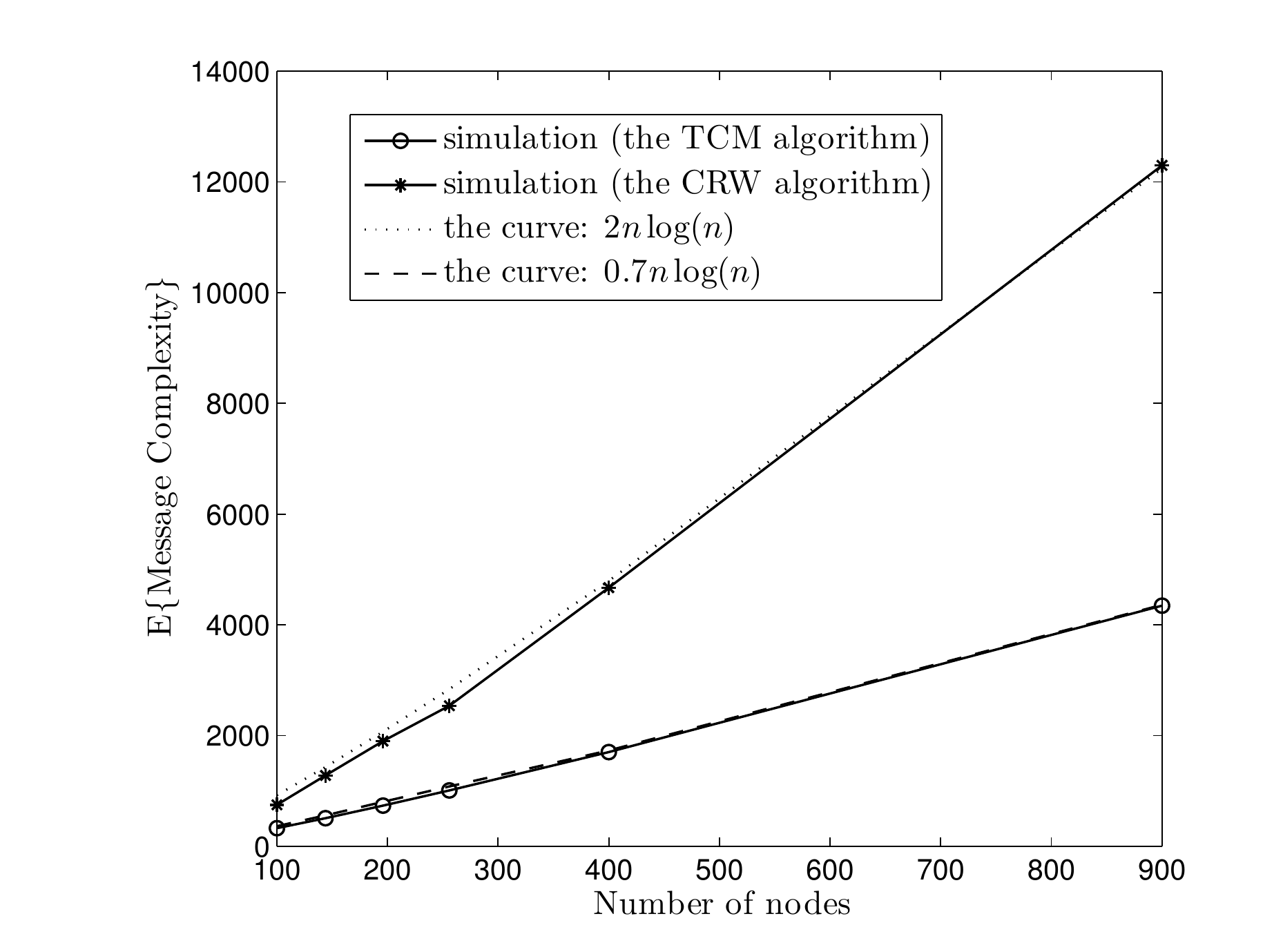}
        \label{fig5b}
    }
	 \caption{Average time and message complexities of TCM and CRW algorithms in torus networks.}
\label{fig5}
\end{figure}

In Fig. \ref{fig5}, we evaluate the average time and message complexities of TCM and CRW algorithms in torus networks. We can see that TCM algorithm has at least a gain of $\log(n)$ in time complexity and a scale factor of $2.85$ in message complexity. In Fig. \ref{fig7n}, the average time and message complexities of TCM and CRW algorithms are depicted in Erd\"{o}s-Renyi model. According to Fig. \ref{fig7na}, the TCM algorithm has an improvement in time complexity by a factor $\sqrt{n}$. Furthermore, the average message complexity of TCM algorithm is approximately half of the CRW algorithm.
\begin{figure}[!t]
    \centering
    \subfigure[Average time complexity]
    {
        \includegraphics[width=3in]{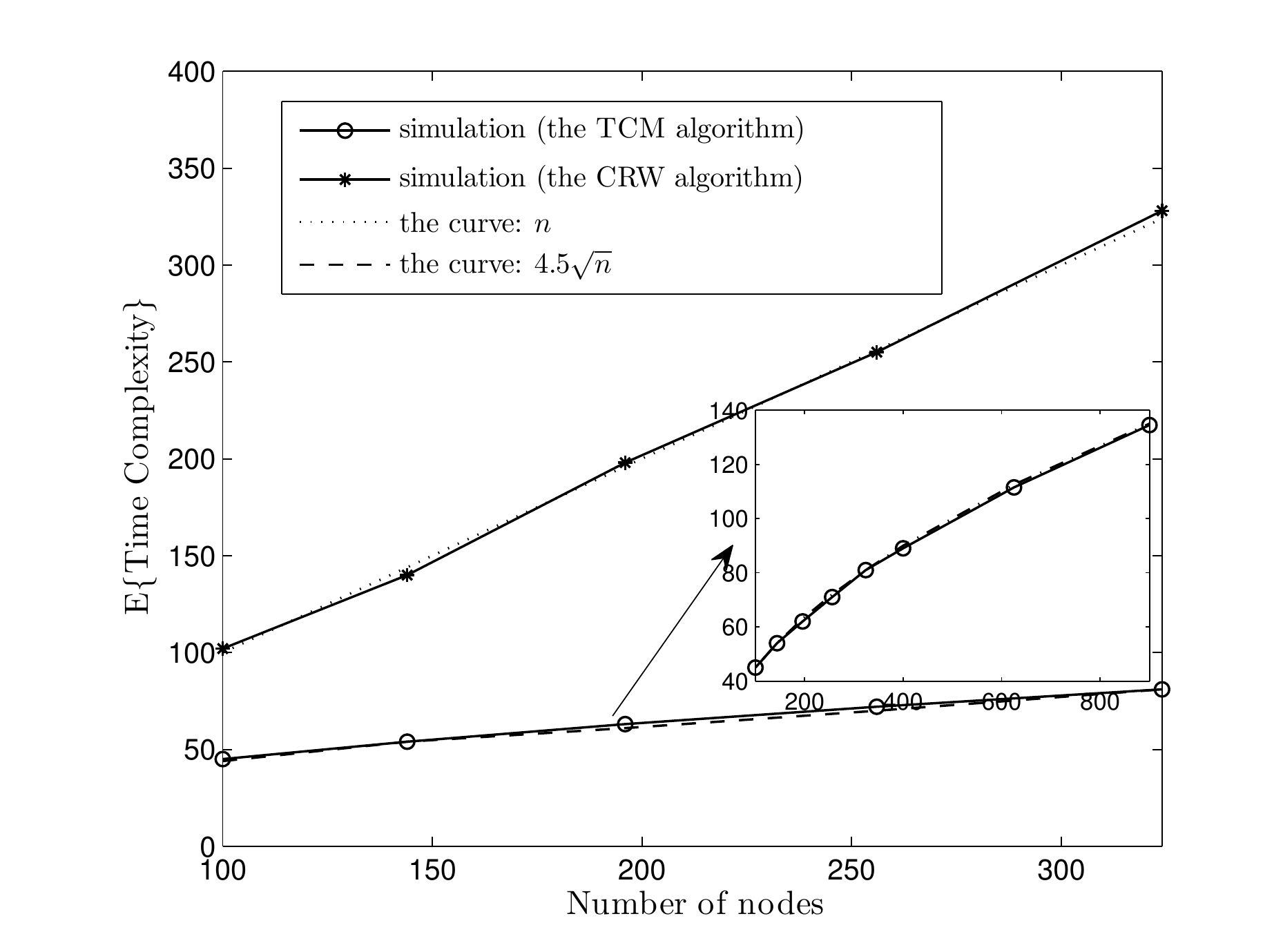}
        \label{fig7na}
    }
    \subfigure[Average message complexity]
    {
        \includegraphics[width=3in]{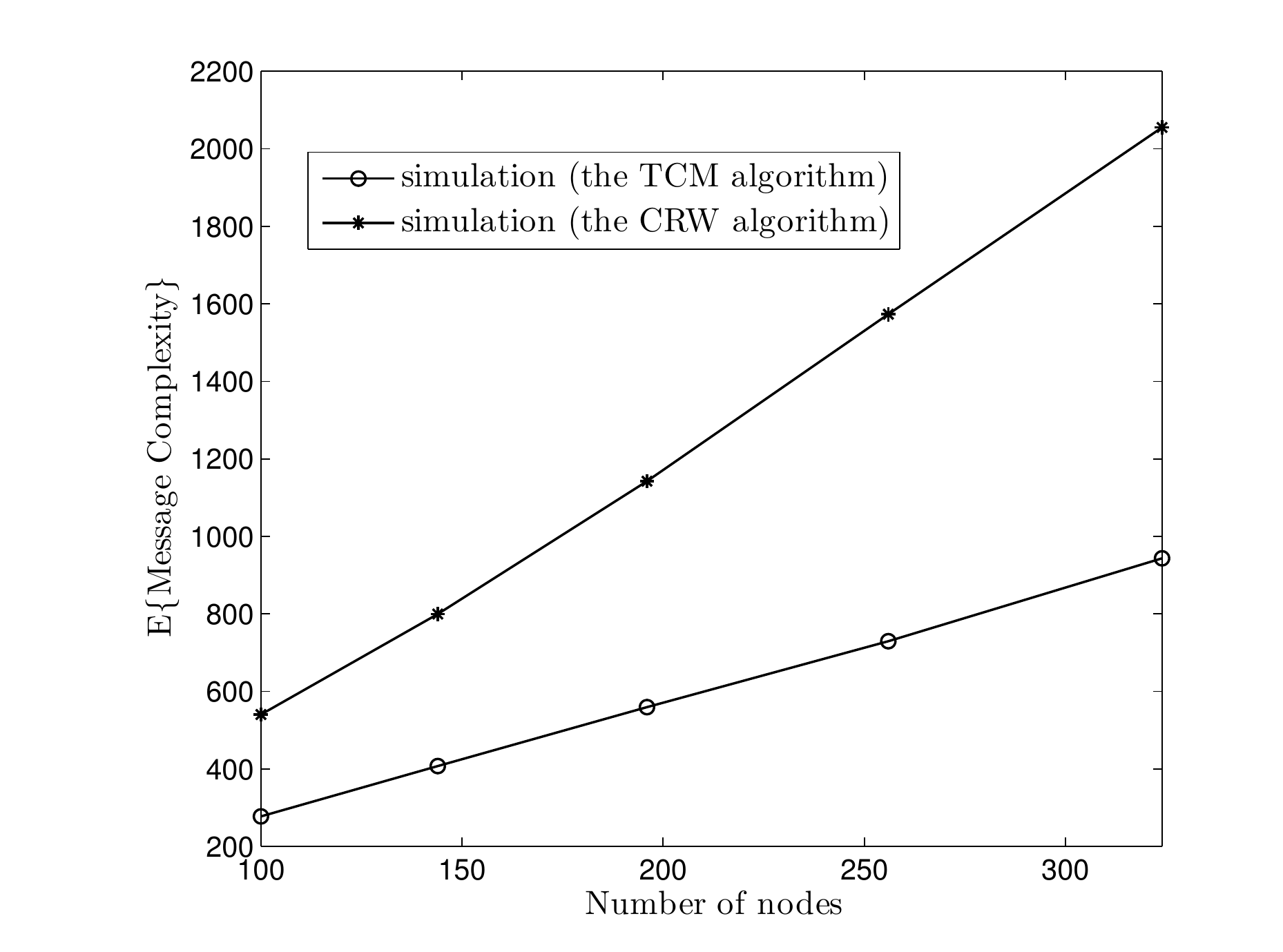}
        \label{fig7nb}
    }
	 \caption{Average time and message complexities of TCM and CRW algorithms in Erd\"{o}s-Renyi model.}
\label{fig7n}
\end{figure}

In Fig. \ref{figpsucc}, the probability of successful computation by running one instance of TCM and CRW algorithms are depicted in the case of complete graphs. The failure rate is set to $0.05/n$. For the TCM algorithm, $P_{succ}$ is approximately equal to $0.83$ for different values of $n$ in the range $[100,400]$. Besides, results from analysis are close to it by an offset of $0.001$. In the case of CRW algorithm, results from the simulation and the analysis are also close to each other. For this algorithm, $P_{succ}$ is greater than $0.74$ for various values of $n$ in the range $[100,400]$. 
\begin{figure}[!t]
    \centering
    \subfigure[The TCM algorithm]
    {
        \includegraphics[width=3in]{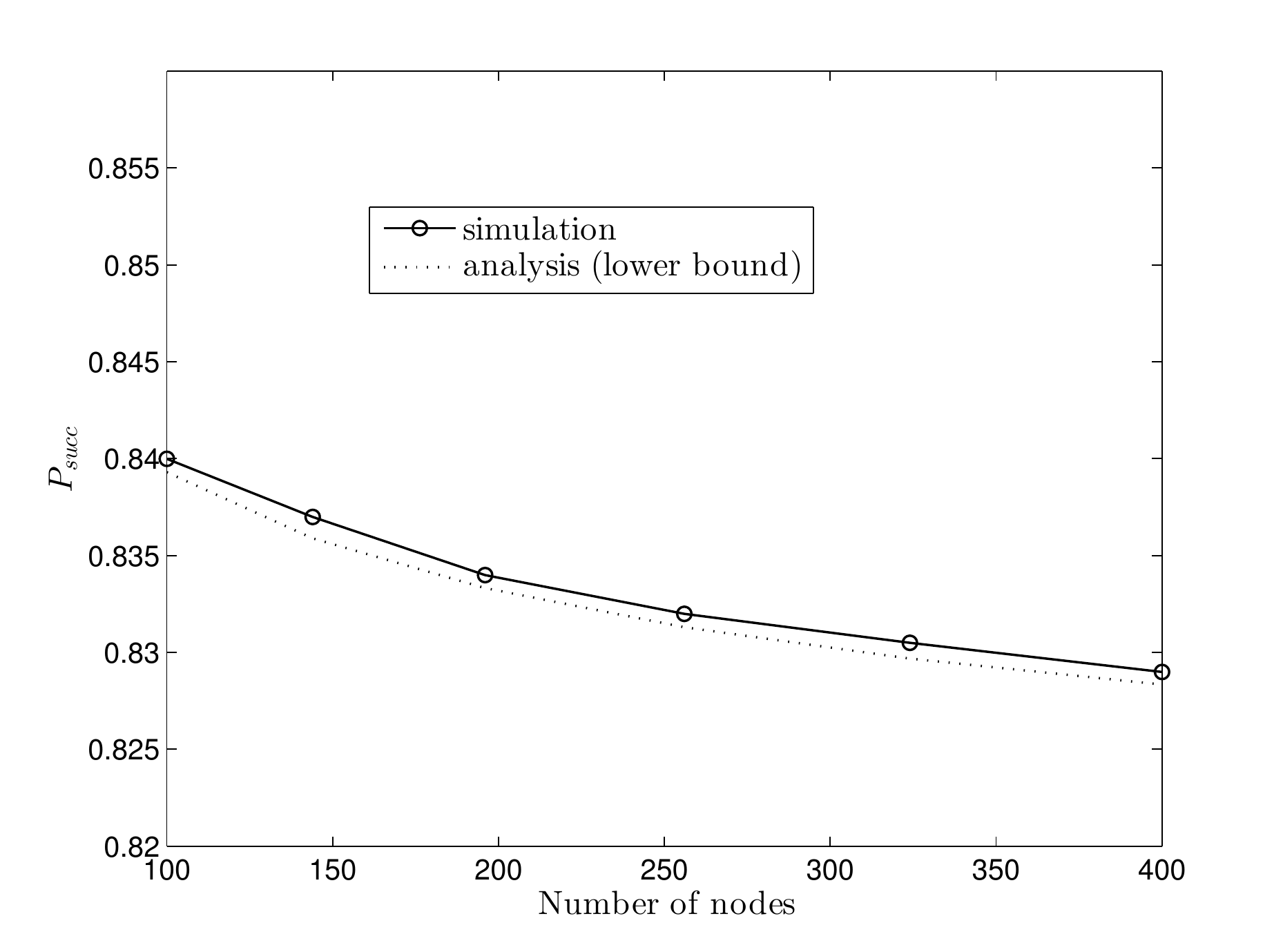}
        \label{figpsucca}
    }
    \subfigure[The CRW algorithm]
    {
        \includegraphics[width=3in]{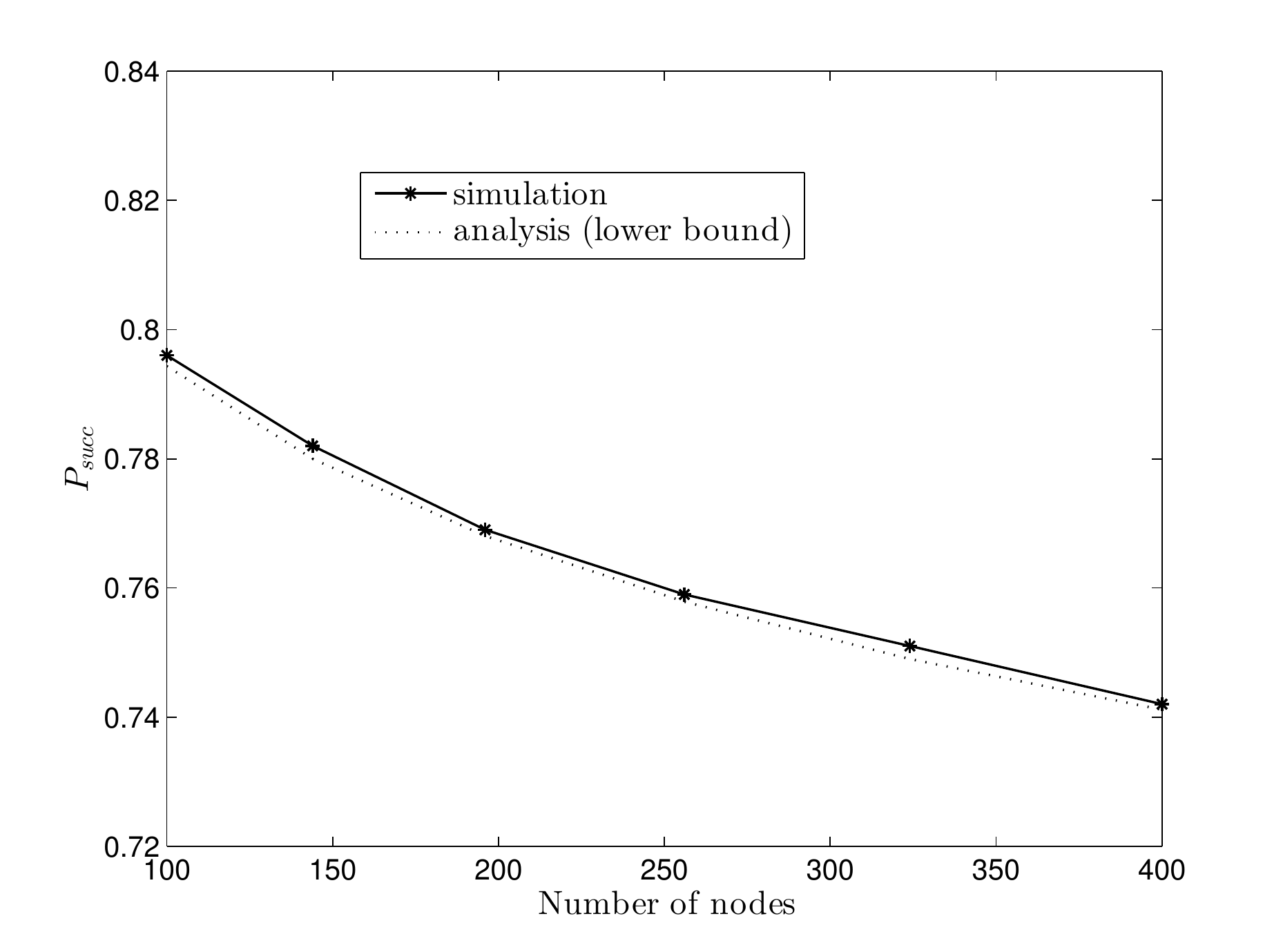}
        \label{figpsuccb}
    }
	 \caption{The probabilities of successful computation in TCM and CRW algorithms for complete graphs, $R=1$.}
\label{figpsucc}
\end{figure} 

\begin{figure}[!t]
    \centering
    \subfigure[Message complexity]
    {
        \includegraphics[width=3in]{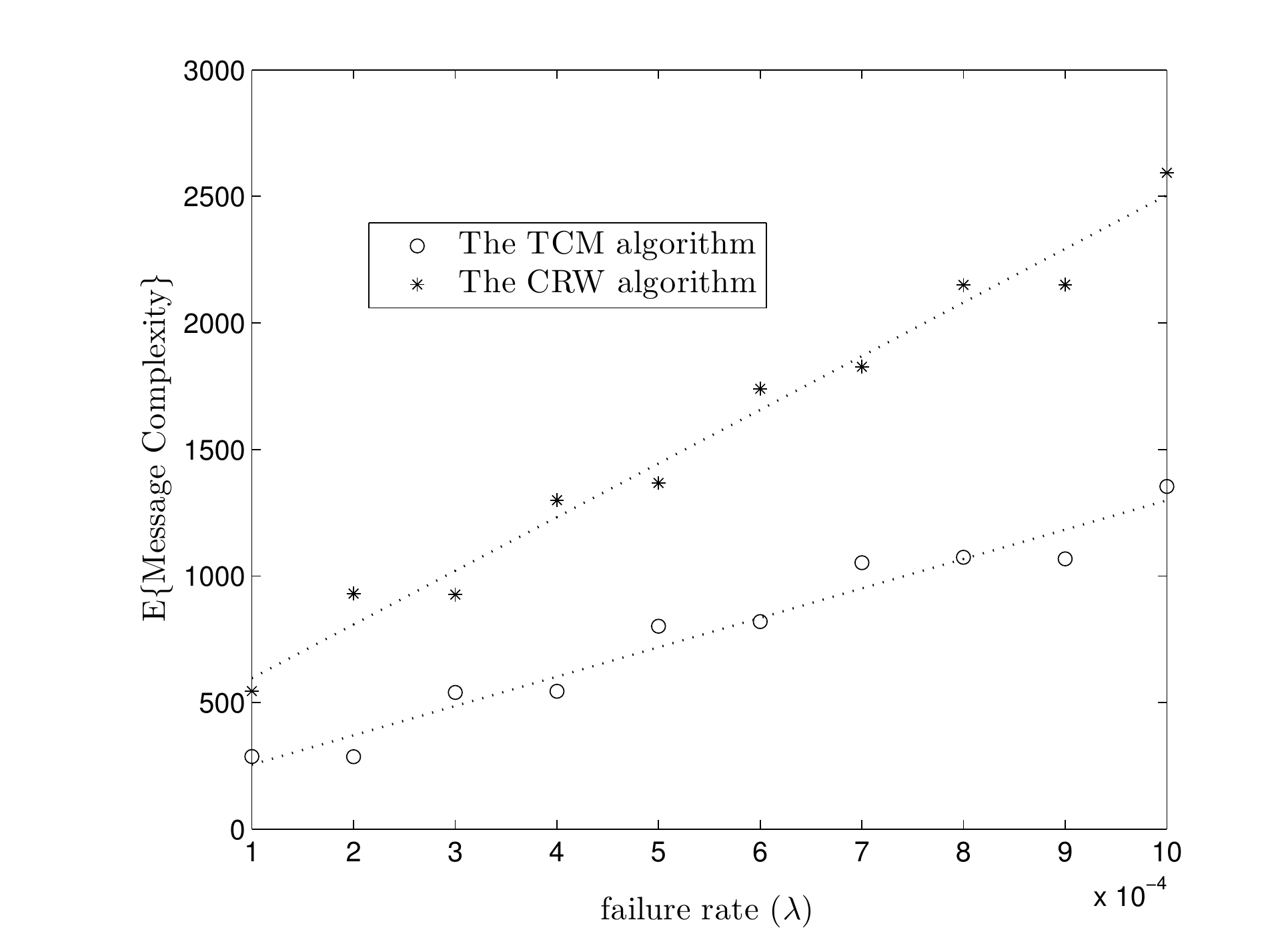}
        \label{figrobust1a}
    }
    \subfigure[Time complexity]
    {
        \includegraphics[width=3in]{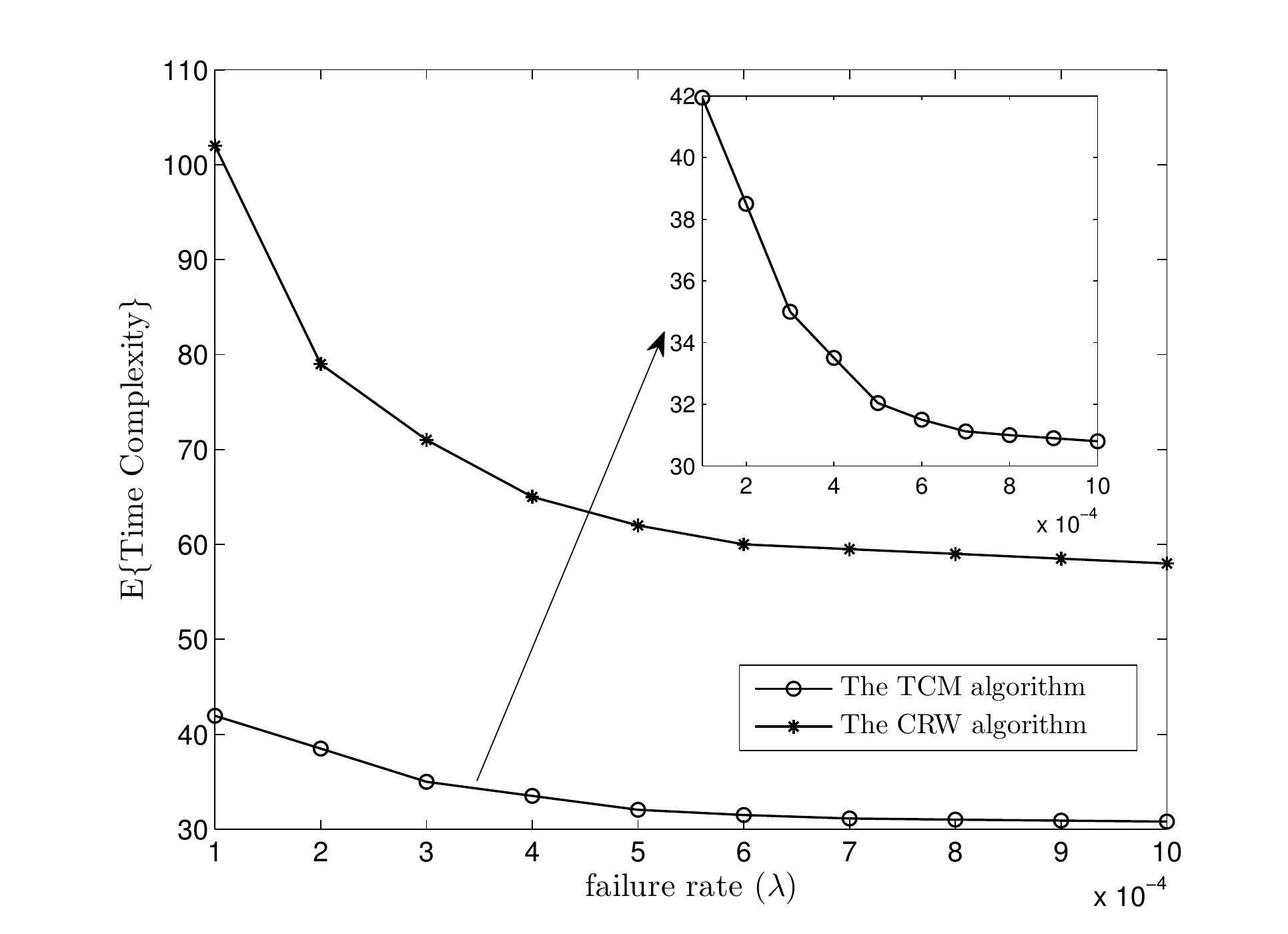}
        \label{figrobust1b}
    }
	 \caption{Average time and message complexities of TCM and CRW algorithms versus failure rate in complete graphs, $n=100$. The dashed lines represent the linear regression between message complexity and failure rate.}
\label{figrobust1}
\end{figure} 

In Fig. \ref{figrobust1a}, the message complexities of the TCM and CRW algorithms are plotted versus failure rate in a complete graph with $n=100$ nodes. The number of parallel instances is determined such that the probability of successful computation is equal to $0.95$. As it can be seen, it is required to run a few more instances of the TCM and CRW algorithms to tolerate higher failure rate. Furthermore, message complexity of the TCM algorithm is less than the one for the CRW algorithm. In Fig. \ref{figrobust1b}, the time complexities of both algorithms are given versus failure rate. For higher failure rate, we need to run more instances of the TCM/CRW algorithm to have $P_{succ}=0.95$. On the other hand, executing multiple instance of the algorithms improves the time complexity. Since the target function is computed if any of the instances is terminated successfully. 

\begin{figure}[!t]
\centering
\includegraphics[width=3.5in]{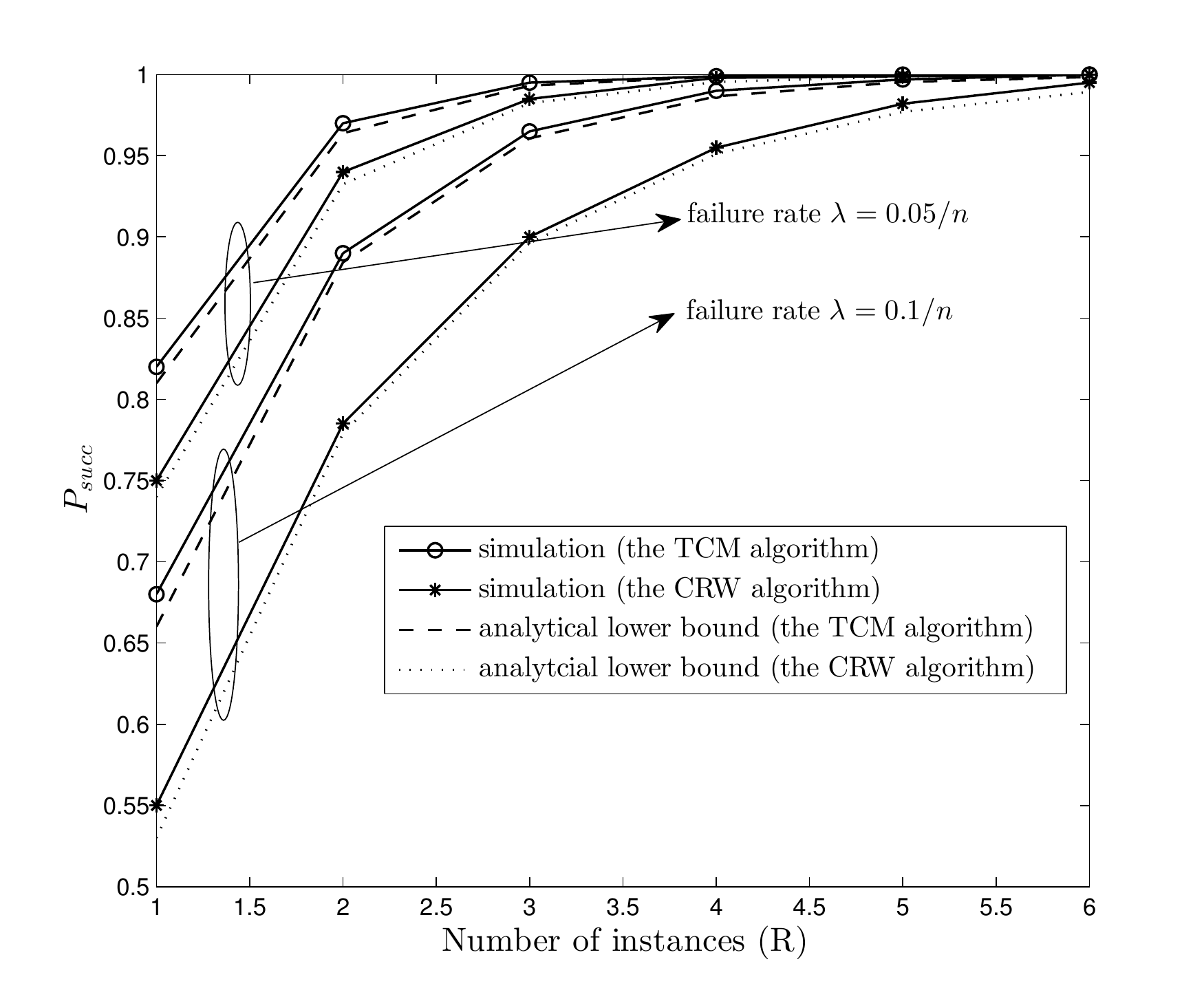}
\caption{The probability of successful computation versus number of multiple instances in complete graphs, $n=400$.}
\label{figrobust2}
\end{figure}

In Fig. \ref{figrobust2}, the probabilities of successful computation of the TCM and CRW algorithms are plotted versus number of multiple instances in a complete graph with $n=400$ nodes for the failure rates $\lambda=0.05/n,0.1/n$. It can be seen that the analytical lower bounds in (\ref{eqrobust1}) and (\ref{eqrobust2}) are close to simulation results. Furthermore, $P_{succ}$ goes to one in all cases when $6$ number of instances are executed in parallel. Thus, the proposed solution makes both algorithms robust against node failures by running a few number of instances in parallel as we expected from Corollaries \ref{col4_1} and \ref{col4_2}.

Studying the impact of dynamic topologies on the performance of distributed algorithms is quite important. Here, we evaluate the performance of TCM and CRW algorithms under node mobility. There exist different mobility models in the literature of mobile ad hoc networks \cite{camp2002survey}. In the simulations, we consider the Random Walk (RW) mobility model which is frequently used in determining the protocol performance and it can mimic movements of mobile nodes walking in an unpredictable way \cite{camp2002survey}. 

Initially, suppose that nodes are located randomly over a square of unit area. Let $[x_i(t),y_i(t)]$ be the location of node $i$ at time $t$. In the RW mobility model, the differences $x(t+h)-x(t)$ and $y(t+h)-y(t)$ are two independent normally distributed random variables with zero mean and variance $2Dh$ $,\forall h>0$ where $D$ is the diffusion coefficient \cite{groenevelt2006relaying}. Thus, the mean square displacement of a node is related to the parameter $D$. In particular, the probability of large displacement increases as diffusion coefficient $D$ grows. We assumed that if a node reaches the boundary of simulated area, it will be bounced off the boundary according to the same angle. Furthermore, two nodes are neighbor if the distance between them is less than a fixed transmission range. The transmission range is set to a value such that the graph remains connected with high probability for the static case, i.e. $D=0$ \cite{shah2009}. 

\begin{figure}[!t]
    \centering
    \subfigure[Time complexity]
    {
        \includegraphics[width=3in]{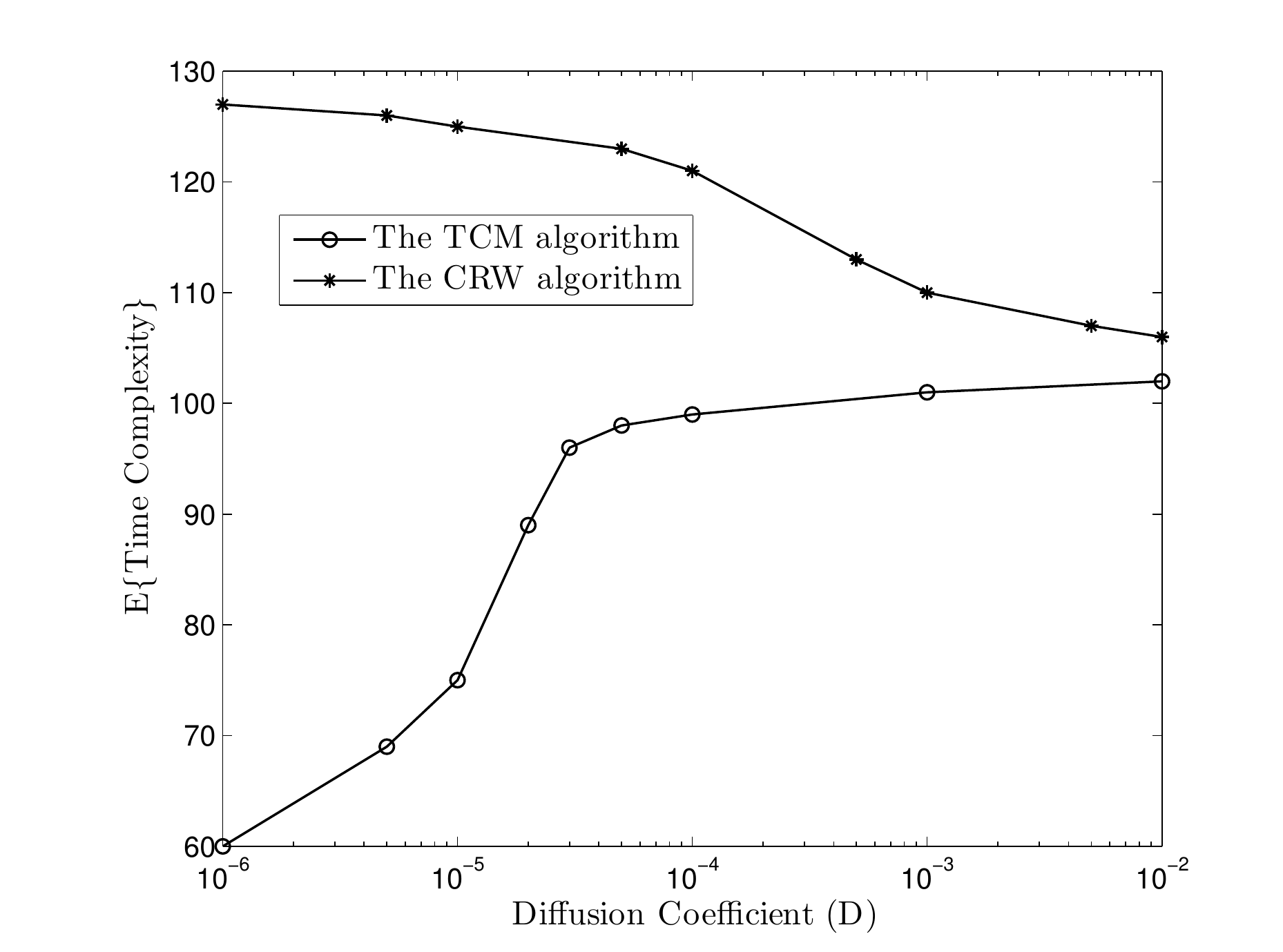}
        \label{figmoba}
    }
    \subfigure[Message complexity]
    {
        \includegraphics[width=3in]{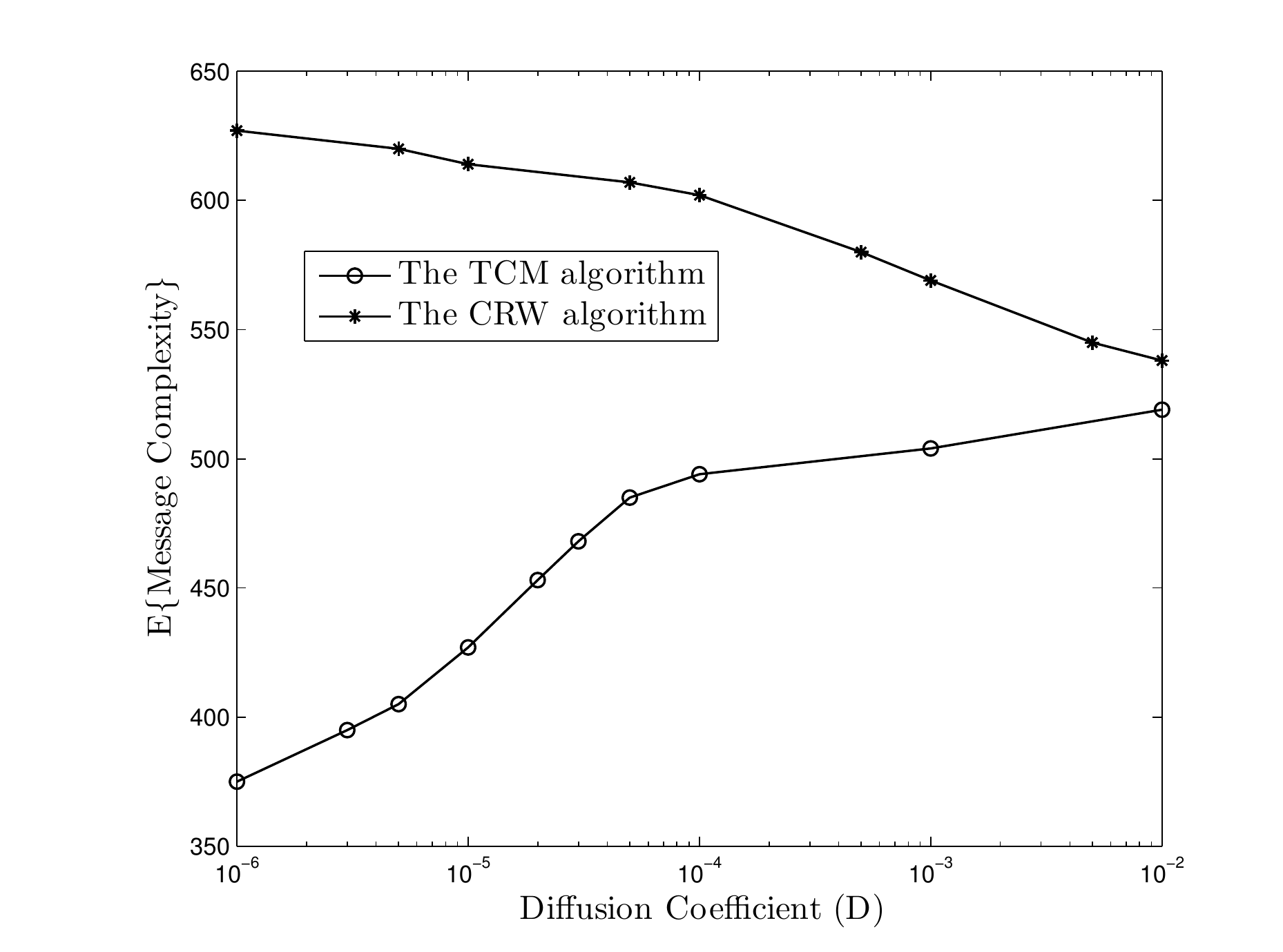}
        \label{figmobb}
    }
	 \caption{Average time and message complexities of TCM and CRW algorithms versus diffusion coefficient $D$ in a network with $n=100$ nodes which are deployed in a square of unit area. The transmission range is set to $0.18$.}
\label{figmob}
\end{figure} 

In the TCM algorithm, we assume that each node $i$ registers the UID of the node that the token $memory(i)$ passed to it. Whenever an active node should send a token to a node which is not in its transmission range any more, it will pass its token to a random neighbor node. In Fig. \ref{figmob}, the time and message complexities of TCM and CRW algorithms are depicted versus the parameter $D$ in a network with $n=100$ nodes. It is noteworthy that both algorithms can compute the class of target functions defined in Lemma \ref{lemma2_1} successfully even in high mobility networks. Furthermore, the time and message complexities of TCM algorithm increases as the parameter $D$ grows while node mobility improves the performance of CRW algorithm. In fact, higher mobility weakens the advantage of chasing mechanism. On the other hand, it gives an opportunity to a completely randomized solution, i.e. the CRW algorithm, to reduce the coalescing time of distant tokens. Nevertheless, simulation results show that the TCM algorithm outperforms the CRW algorithm in both time and message complexities.  

\section{Conclusions}
In this paper, we proposed the TCM algorithm to compute a wide class of target functions (such as sum, average, min/max, XOR) in a distributed manner. In complete graph and Erd\"{o}s-Renyi model, we showed that it reduces running time at least by factor $\sqrt{n/\log(n)}$ with respect to completely randomized solution, i.e. the CRW algorithm, and there is at least a factor of $\log(n)/\log(\log(n))$ improvement in torus networks. We defined a robustness metric to study the impact of node failures on the performance of CRW and TCM algorithms. The TCM and CRW algorithms can tolerate the failure rate of $\alpha/n$ by running $O(n^\alpha)$ and $O(1)$ instances in parallel, respectively. Furthermore, simulation results showed that both algorithm can compute the target functions successfully even in high mobility conditions. 

\section{ Appendix A}

{\em  Proof of Lemma III.1}: 

The pdf of $|EH_1(k)|$ can be approximated with Gaussian distribution $\mathcal{N}(\mu_k,\sigma_k)$ where $\mu_k=n-n(1-1/n)^k$ and $\sigma_k^2=n^2(1-1/n)(1-2/n)^k+n(1-1/n)^k-n^2(1-1/n)^{2k}$ [29]. After some manipulations, we have: 
\begin{equation}
\Pr\{|EH_1(2k)|\leq \mathbb{E}\{|EH_1(k)|\}\} \leq \frac{1}{2}e^{-(\mu_{k}-\mu_{2k})^2/2\sigma_{2k}^2} \leq e^{-n/4-k\eta}.
\end{equation}
where $\eta\geq 0.05$. Hence, the size of the set $EH_1(2k)$, is greater than $\mathbb{E}\{|EH_1(k)|\}$ with probability at least $1-e^{-n/4 -k\eta }$. 
\section{ Appendix B}
{\em Proof of Theorem III.2}: 

Consider token $ID_i$ ($i>1$). Let $x_k^i$ be the node visited by token $ID_i$ at $k$-th step and $S_k^i=\{x_1^i,\cdots,x_k^i\}$ be the history of the corresponding walk. We define the walk taken by token $ID_i$ as {\em weakly self-avoiding walk}, provided that: 
\begin{equation}
\Pr\{x_{k+1}^i|S_k^i\}
\begin{cases}
= \alpha_k  \qquad x_{k+1}^i\not \in S_k^i,\\
\leq \alpha_k \qquad  x_{k+1}^i\in S_k^i,
\end{cases}
\label{self}
\end{equation}
for some $\alpha_k$ where $\alpha_k\geq \frac{1}{n-1}$. Thus, in a weakly self-avoiding walk, token $ID_i$ visits new nodes with higher probability than the visited nodes.

\begin{mylm} In the TCM algorithm, the path traced by token $ID_i$ ($i>1$) is a weakly self-avoiding walk.
\label{lemmaa_1}
\end{mylm}
\begin{proof}
Suppose that the token $ID_i$ enters a node $x_k^i$ visited by some other token with higher UID. Let $ID_{i^{\prime}}$ be the maximum UID, node $x_k^i$ has seen so far. Furthremore, assume that token $ID_{i^{\prime}}$ is in $k^{\prime}$ steps and has visited node $x_k^i$ in $j$-th step for the last time, i.e. $j=\max_{\omega\leq k^{\prime}} \omega \mbox{ } s.t. \mbox{ } x_{\omega}^{i^{\prime}}= x_k^i$. We denote the chasing and random walk modes of token $ID_i$ by $chase_i$ and $RW_i$, respectively. Now, for a given history $S_k^i$, we have:
\begin{equation}
\begin{split}
\Pr\{x_{k+1}^i=a|S_k^i\}&=\displaystyle\sum_{i^{\prime}=1}^{i-1} \Big[\displaystyle\sum_{j=1}^{k^{\prime}} \Pr\{x_{k+1}^i=a|x_j^{i^{\prime}}=x_k^i,chase_i,S_k^i\}\times \Pr\{x_j^{i^{\prime}}=x_k^i,chase_i|S_k^i\}\Big]\\   & \qquad \qquad+\Pr\{x_{k+1}^i=a|RW_i,S_k^i\}\times \Pr\{RW_i|S_k^i\}. 
\label{eql1}
\end{split}
\end{equation}

Suppose that token $ID_i$ was in $l$-th step when token $ID_{i^{\prime}}$ was leaving node $x_j^{i^{\prime}}$ (see Fig. \ref{figsupp1}). We prove that token $ID_i$ will not visit nodes in the set $\{x_l^i,\cdots,x_k^i\}$ in the next step. By contradiction, assume that there exists $l\leq p \leq k$ where $x_p^i=x_{j+1}^{i^{\prime}}$. However, we have:
\begin{equation}
\forall \omega_1\in \{p,\cdots,k\}, \exists \omega_2 \in\{j+1,\cdots,k^{\prime}\}, \mbox{ } s.t. \mbox{ } x_{\omega_1}^i=x_{\omega_2}^{i^{\prime}},
\end{equation}
due to the fact that token $ID_i$ is chasing token $ID_{i^{\prime}}$. For $\omega_1=k$, the above equation asserts that token $ID_{i^{\prime}}$ revisited node $x_k^i$ in some step later than $j$ which is contradiction. 

We know that token $ID_{i^{\prime}}$ was eventually in the random walk mode in $j$-th step. Hence, each node in the set $\{1,\cdots,n\}\backslash \{x_l^i,\cdots,x_k^i\}$ is selected with probability $1/(n-|\{x_l^i,\cdots,x_k^i\}|)$ in the $k+1$-th step. Consequently, we have:
\begin{equation}
\begin{split}
\Pr\{x_{k+1}^i=a|x_j^{i^{\prime}}=x_k^i,chase_i,S_k^i\}&\geq \Pr\{x_{k+1}^i=b|x_j^{i^{\prime}}=x_k^i,chase_i,S_k^i\}\\ &\forall j,k, \forall a,b\in \{1,\cdots,n\} , a\not\in S_k^i, b\in S_k^i.
\end{split}
\label{eql2}
\end{equation}
From (\ref{eql1}) and (\ref{eql2}), it can be concluded that:
\begin{equation}
\Pr\{x_{k+1}^i=a|S_k^i\} \geq \Pr\{x_{k+1}^i=b|S_k^i\}, \forall a,b\in \{1,\cdots,n\} , a\not\in S_k^i, b\in S_k^i.
\end{equation}

Thus, the proof is complete.

\end{proof}
 
 \begin{figure}[!t]
\centering
\includegraphics[width=5in]{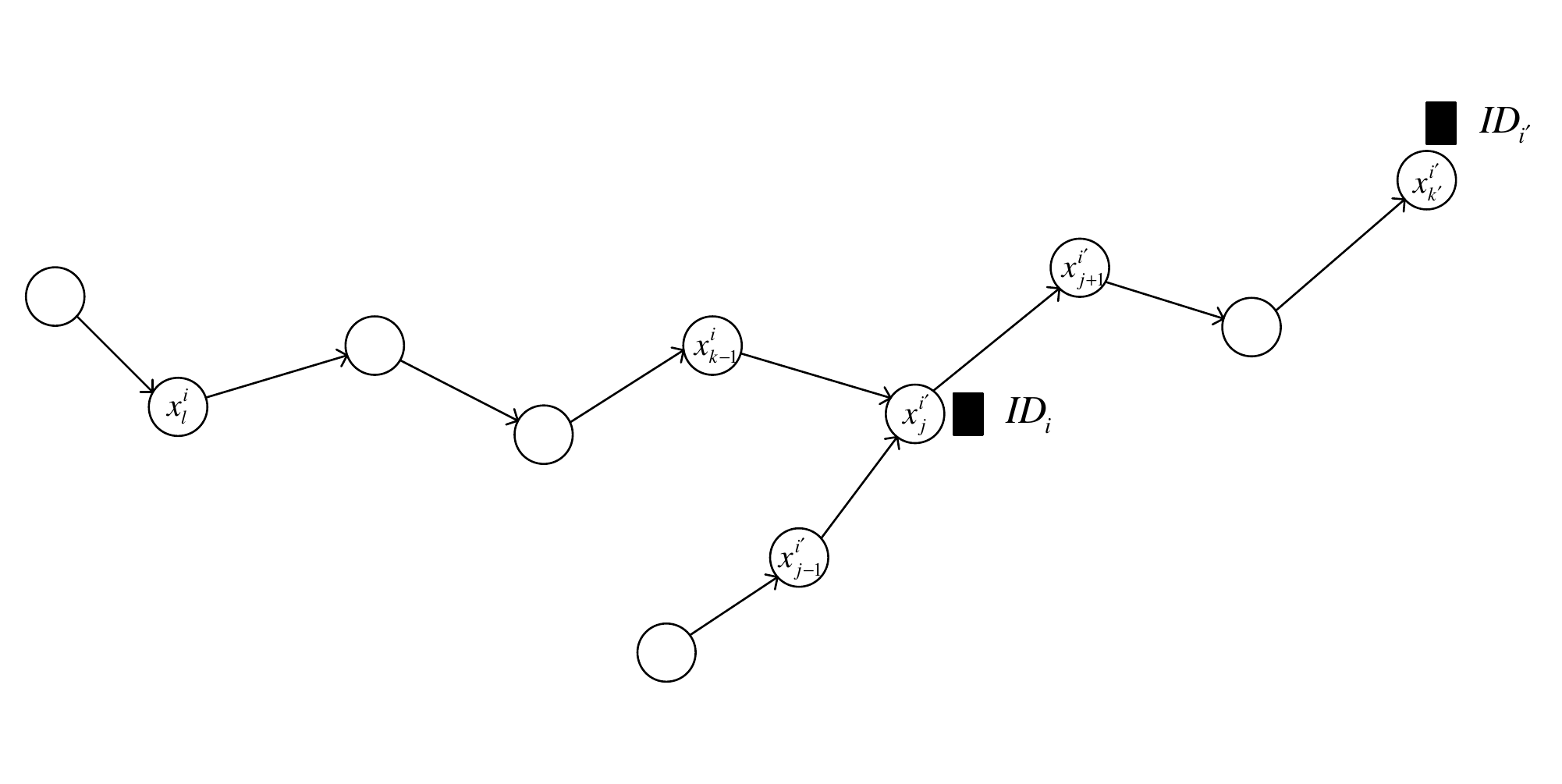}
\caption{Token $ID_i$ visits node $x_{k}^i=x_j^{i^{\prime}}$. Thus, token $ID_i$ will not visit nodes $\{x_l^i,\cdots,x_k^i\}$ in the next step. Tokens are depicted by black squares.}
\label{figsupp1}
\end{figure}

Assume that if token $ID_i$ coalesces with token $ID_j$ (where $ID_j>ID_i$), it virtually sticks to token $ID_j$. Now, if token $ID_j$ meets another token, say $ID_k$ with higher UID, token $ID_j$ and all tokens attached to it, stick to token $ID_k$. This process continues until token $ID_i$ hits the event horizon of $ID_1$ by itself or another token. We denote the time for token $ID_i$ to hit the event horizon of token $ID_1$ by $T_{EH1}(ID_i)$. Furthermore, let $EH_i(t)$ be the set of nodes visited by token $ID_i$ up to time $t$. 

Token $ID_1$ takes steps in the network according to a Poisson process with rate $1/2$ (assuming that $p_{send}=1/2$). At each step, it chooses one of nodes except its current node with probability $1/(n-1)$. Thus, each node (excluding the initial node having token $ID_1$) is not visited by token $ID_1$ up to time $t$ with probability $e^{-t/2(n-1)}$ independently from other nodes. Hence, the pdf of the number of visited nodes at time $t$ is:
\begin{equation}
\Pr\{|EH_1(t)|=r\}={n-1 \choose r-1} (1-e^{-t/2(n-1)})^{r-1}\times (e^{-t/2(n-1)})^{n-r},
\label{eqa2}
\end{equation}
for $1\leq r \leq n-1$.

\begin{mylm} We have the following probabilistic bound on the number of visited nodes by token $ID_1$ at time $2t$:
\begin{equation}
\Pr\{|EH_1(2t)|\leq \mathbb{E}\{|EH_1(t)|\}\} \leq  e^{-\alpha_0 t}, t\leq 2n,
\end{equation}
where $\alpha_0=(1-\log(2))/4$.
\label{lemmab1}
\end{mylm}
\begin{proof}
From (\ref{eqa2}) and the proposed upper bound for binomial distribution in [30], we have:
\begin{equation}
\Pr\{|EH_1(2t)|\leq \mathbb{E}\{|EH_1(t)|\}\} \leq  e^{-nD},
\end{equation}
where $D=a\log(a/b)+(1-a)\log((1-a)/(1-b))$, $a=1-e^{-t/2(n-1)}$ and $b=1-e^{-t/(n-1)}$. Besides, we have:
\begin{equation}
\begin{split}
nD&=t/2e^{-t/2(n-1)}-(n-1)(1-e^{-t/2(n-1)})(1+e^{-t/2(n-1)})
\\&>t/2(1-\log(2))+t^2(\log(2)-1)/(8n).
\end{split}
\end{equation}

From above equation, it can be easily seen that $nD>t(1-\log(2))/4$ for $t\leq 2n$. Therefore, the proof is complete.
\end{proof}

\begin{mylm} Let $N_i(t_0,t_0+2t)$ be the number of steps taken by token $ID_i$ in time interval $[t_0,t_0+2t]$. Then, we have the following bound:
\begin{equation}
\Pr\{N_i(t_0,t_0+2t)<\lfloor t/2\rfloor\}\leq e^{-\alpha_1t},
\end{equation}
where $\alpha_1=\log(\sqrt{e/2})$.
\label{lemmab2}
\end{mylm}
\begin{proof}
The random variable $N_i(t_0,t_0+2t)$ is a Poisson process with rate at least $\lambda_{2t}=2t\times 1/2=t$. Thus, we have from the Chernoff bound:
\begin{equation}
\Pr\{N_i(t_0,t_0+2t)\leq \lfloor t/2\rfloor\}\leq\displaystyle\sum_{i=0}^{\lfloor t/2\rfloor} \frac{e^{-\lambda_{2t}}(\lambda_{2t})^i}{i!}\leq e^{-t}\frac{(et)^{t/2}}{(t/2)^{t/2}}=(2/e)^{t/2}.
\end{equation}
The proof is complete.
\end{proof} 
\begin{myrem}
By the same arguments in Lemma \ref{lemmab2}, it can be shown that: $\Pr\{ N_i(t_0,t_0+t)>2t\}\leq e^{-\alpha_2 t}$ where $\alpha_2=\log{(4/e)}$.
\label{remarkB.1}
\end{myrem}

Given a time $t$, we say that the event $E_i(t)$ occurs if $|EH_1(t)\backslash EH_i(t)|\geq 1/8\sqrt{n\log{n}}$. Let $E(t)=\underset{i}{\bigcap} E_i(t)$ and define $t^{\star}=\sqrt{n\log{n}}$. We have:
{\small
\begin{equation}
\begin{split}
\Pr\big\{ E^c(t^{\star})\big\}&\leq^a \displaystyle\sum_{i=2}^n\Pr\{E_i^c(t^{\star})\} \leq (n-1) \mathbb{E}\Bigg\{\displaystyle\sum_{j= |EH_1(t^{\star})|-1/8\sqrt{n\log{n}}}^{N_i(0,t^{\star})} { N_i(0,t^{\star}) \choose j} P_{EH_1(t^{\star})}^j P_{EH^c_1(t^{\star})}^{N_i(0,t^{\star})-j}\Bigg\} \\
&\leq^b \mathbb{E}\Bigg\{(n-1) \displaystyle\sum_{j=|EH_1(t^{\star})|-1/8\sqrt{n\log{n}}}^{N_i(0,t^{\star})} { N_i(0,t^{\star}) \choose j} \Big(\frac{|EH_1(t^{\star})|}{n-N_i(0,t^{\star})}\Big)^j\Bigg\}\\ 
&\leq^c (n-1) \Bigg(\displaystyle\sum_{j=\lfloor 1/8\sqrt{n\log{n}}\rfloor}^{\lceil 2\sqrt{n\log{n}}\rceil} { 2\sqrt{n\log{n}} \choose j} \Big(\frac{1/4\sqrt{n\log{n}}}{n-2\sqrt{n\log{n}}}\Big)^j+ e^{-\alpha_0 \sqrt{n\log{n}}/2} + e^{-\alpha_2 \sqrt{n\log{n}}}\Bigg)\\
&\leq^d \frac{1}{\sqrt{n\log{n}}}.
\end{split}
\label{eqkey}
\end{equation}
}
\\(a) The first sum is given according to the union bound. The second sum is greater than the probability of having $|EH_1(t^{\star})\cap EH_i(t^{\star})|\geq j$ where $P_{EH_1(t^{\star})}$ and $P_{EH^c_1(t^{\star})}$ are the probabilities of choosing a node from the set $EH_1(t^{\star})$ and $\{1,\cdots,n\}\backslash EH_1(t^{\star})$, respectively.
\\(b) From Lemma \ref{lemmaa_1}, the path traced by token $ID_i$ ($i>1$) is a weakly self-avoiding walk. Thus, we have: $P_{EH_1(t^{\star})}\leq \frac{|EH_1(t^{\star})|}{n-N_i(0,t^{\star})}$. 
\\(c) The sum has greater value for larger $|N_i(0,t)|$ and smaller $|EH_1(t)|$. We can obtain this inequality by bounding the probability $\Pr\{|EH_1(t^{\star})|<1/4\sqrt{n\log{n}}\}$ and $\Pr\{N_i(0,t^{\star})>2\sqrt{n\log{n}}\}$ from Lemma \ref{lemmab1} and Remark \ref{remarkB.1}, respectively.
\\(d) From Strling's approximation, the probability is in the order of $O(e^{-\log{n}\sqrt{n\log{n}}})$. Thus, it is less than $1/\sqrt{n\log{n}}$ for large enough $n$. 

\begin{mylm} Assume that the event $E(t^{\star})$ occurs. Then, the probability of not hitting the event horizon of token $ID_1$ by token $ID_i$ after $t^{\star}+2t$ is less than the following:
\begin{equation}
\Pr\{T_{EH1}(ID_i)>t^{\star}+2t|E(t^{\star})\}\leq e^{-\frac{1}{16} \sqrt{\log{n}/n}t} + e^{-\alpha_1 t}.
\end{equation}
\label{lemmab3}
\end{mylm}
\begin{proof}
Suppose that the event $E(t^{\star})$ occurs at time $t^{\star}$. Thus, the size of the the set $EH_1(t)\backslash EH_i(t)$, $t>t^{\star}$, will be greater than $1/8\sqrt{n\log(n)}$ as far as token $ID_i$ does not hit it. Hence, the probability of not hitting the event horizon of $ID_1$ in time interval $[t^{\star},t^{\star}+2t]$ is less than $(1-1/8\sqrt{n\log{n}}/n)^{N_i(t^{\star},t^{\star}+2t)}$. By bounding $N_i(t^{\star},t^{\star}+2t)$ from below (see Lemma \ref{lemmab2}), we have:
\begin{equation}
\begin{split}
\Pr\{ T_{EH1}(ID_i)>t^{\star}+2t| E(t^{\star})\} & \leq \Pr\{ N_i(t^{\star},t^{\star}+2t)>\lfloor t/2\rfloor\} \times (1-1/8\sqrt{n\log{n}}/n)^{t/2} \\
&\quad + \Pr\{N_i(t^{\star},t^{\star}+2t)\leq \lfloor t/2\rfloor\}\times 1,\\
& \leq e^{-\frac{1}{16} \sqrt{\log{n}/n}t} + e^{-\alpha_1 t}.
\end{split}
\end{equation}

\end{proof}

\begin{mylm} Suppose that token $ID_i$ hits the event horizon of token $ID_1$ at time $t$. Then, it will coalesce with token $ID_1$ in next $3t$ time units with probability greater than $1-(e^{-\alpha_3 t}+e^{-\alpha_4 t})$ where $\alpha_3=1/36$ and $\alpha_4=\log(2/\sqrt{e})$.
\label{lemmab4}
\end{mylm}
\begin{proof}
In worst case scenario, the event horizon of token $ID_1$ is a line with length $N_1(0,t)$ and token $ID_i$ hits end of the line at time $t$. Thus, token $ID_i$ reaches token $ID_1$ at time $t^{\prime}$ given in the following equation:
\begin{equation}
N_i(t,t^{\prime})=N_1(0,t)+N_1(t,t^{\prime}).
\end{equation}

 Let us define random variable $Y(t^{\prime})=N_i(t,t^{\prime})-N_1(t,t^{\prime})$, which is the difference of two independent Poisson random variables $N_i(t,t^{\prime})$ and $N_1(t,t^{\prime})$ with rates $(t^{\prime}-t)$ and $(t^{\prime}-t)/2$, respectively. Hence, the random variable $Y(t^{\prime})$ has Skellam distribution and we have:
\begin{equation}
\Pr\{Y(t^{\prime})< N_1(0,t)\}\leq e^{-\frac{(N_1(0,t)-1/2(t^{\prime}-t))^2}{3(t^{\prime}-t)}}.
\end{equation}

Since token $ID_1$ takes at most $\lceil t \rceil$ steps in time interval $[0,t]$ with probability $\displaystyle\sum_{i=0}^{\lceil t\rceil} e^{-t/2} \frac{(t/2)^i}{i!}\leq e^{-\alpha_4 t}$, we have:
\begin{equation}
\Pr\{Y(4t)< N_1(0,t)\}\leq e^{-\alpha_3 t}+e^{-\alpha_4 t}.
\end{equation}
\end{proof}

\begin{mycol} From Lemmas \ref{lemmab3} and \ref{lemmab4}, we have: 
\begin{equation}
\begin{split}
\Pr\{T_{coal}(ID_i)>4t^{\star}+8t| E(t^{\star})\}&\leq  e^{-\frac{1}{16} \sqrt{\log{n}/n}t}+e^{-\alpha_1 t}+e^{-\alpha_3 t}+e^{-\alpha_4 t},\\
&\leq e^{-\frac{1}{16} \sqrt{\log{n}/n}t}+ 3e^{-\alpha_3 t}.
\end{split}
\end{equation}
\label{corb2}
\end{mycol}

Now, we can obtain an upper bound on the average time complexity:
\begin{equation}
\begin{split}
\mathbb{E}\{T_{run}(n)\}&= \mathbb{E}\{T_{run}(n)| E(t^{\star})\} \Pr\{E(t^{\star})\}+ \mathbb{E}\{T_{run}(n)| E^c(t^{\star})\} \Pr \{ E^c(t^{\star})\}, \\
&\leq^a \mathbb{E}\{T_{run}(n)| E(t^{\star})\} + (4 n\log(n)+2t^{\star}) \times \frac{1}{\sqrt{n\log{n}}},
\\ &=\int_0^{\infty} \Pr\{T_{run}(n)>\tau| E(t^{\star}) \} d\tau + 4 \sqrt{n\log{n}}+2,
\\&  \leq^b  \int_0^{\infty} \min(1,\displaystyle\sum_{i\in \{2,\cdots,n\}} \Pr\{T_{coal}(ID_i)>\tau|E(t^{\star})\}) d\tau + 4\sqrt{n\log{n}}+2,
\\&\leq^c 4t^{\star} +\int_{0}^{\infty} \min(1,(n-1)\times (e^{-\frac{1}{128} \sqrt{\log{n}/n}\tau}+3e^{-\alpha_3\tau/8}) )d\tau + 4\sqrt{n\log{n}}+2,
\\&\leq^d \int_{0}^{128\sqrt{ n \log(n)}} 1 dt +\int_{128\sqrt{n \log(n)}}^{\infty} n\times (e^{-\frac{1}{128} \sqrt{\log{n}/n}\tau}+3e^{-\alpha_3 \tau/8})d\tau + 8\sqrt{n\log{n}}+2,
\\ &\leq 128\sqrt{ n \log(n)} + 128\sqrt{ n/\log{n}} + \frac{24n}{\alpha_3} e^{-16\alpha_3 \sqrt{n\log{n}}}+8\sqrt{n\log{n}}+2 = O(\sqrt{n\log{n}}).
\end{split}
\end{equation}
\\(a) Regardless of the event $E(t^{\star})$, token $ID_1$ covers the complete graph in $t^{\star}+2n\log(n)$ time units on average [13]. Thus, any token $ID_i$ ($i>1$) will coalesce with it in at most $2\times(2n\log{n}+t^{\star})$ time units on average. Hence, we have: $\mathbb{E}\{T_{run}(n)| E^c(t^{\star})\}\leq 4n\log{n} + 2t^{\star}$. Besides, we know that  $\Pr \{ E^c(t^{\star})\}\leq 1/\sqrt{n\log{n}}$ according to (\ref{eqkey}).
\\(b) According to union bound.
\\(c) From Corollary \ref{corb2}.
\\(d) From the fact that $ne^{-\frac{1}{128} \sqrt{\log{n}/n}t}\geq 1$ for $t\leq 128\sqrt{n\log{n}}$.

\bibliographystyle{IEEEtran}
\bibliography{mybib}
\begin{spacing}{1.2}
\noindent
{\footnotesize [29] \space H.-K. Hwang and S. Janson, ``Local limit theorems for finite and infinite urn models,'' {\em The Annals of Probability}, pp. 992-1022, 2008.
}
\\{\footnotesize
[30] \space R. Arratia and L. Gordon, ``Tutorial on large deviations for the binomial distribution,'' {\em Bulletin of mathematical biology},
vol. 51, no. 1, pp. 125-131, 1989.
}
\end{spacing}

\end{document}